\documentclass[11pt]{article}

\usepackage{ifxetex}
\ifxetex
  \usepackage{fontspec}
\else
  \usepackage[T1]{fontenc}
  \usepackage[utf8]{inputenc}
  \usepackage{lmodern}
\fi

\usepackage[margin=1in]{geometry}
\usepackage{graphicx}
\usepackage{xparse}
\usepackage{amsthm}
\usepackage{amssymb}
\usepackage{amsmath}
\usepackage{authblk}
\usepackage{bbm}
\usepackage{marvosym} 
\usepackage{mathtools}
\usepackage{multirow}
\usepackage{thm-restate}
\usepackage{enumitem}
\usepackage{tikz}
\usepackage{enumitem}
\usepackage[binary-units]{siunitx}
\usepackage{subcaption}
\usepackage[colorlinks,citecolor=blue,linkcolor=blue,hypertexnames=false]{hyperref}
\usepackage{hypcap} 
\usepackage[capitalise,noabbrev]{cleveref}
\usepackage[textsize=footnotesize]{todonotes}
\presetkeys{todonotes}{inline}{}
\usepackage{thmtools}
\usepackage{thm-restate}
\usepackage{xspace}
\usepackage{setspace}

\DeclareDocumentCommand\R{}{\mathbb{R}}
\DeclareDocumentCommand\Z{}{\mathbb{Z}}

\DeclareDocumentCommand\conv{o}{\operatorname{conv}\IfValueTF{#1}{\left(#1\right)}{}}
\DeclareDocumentCommand\setdef{mo}{\left\{#1\IfNoValueTF{#2}{}{ \mid #2}\right\}}
\DeclareDocumentCommand\bigsetdef{mo}{\left\{#1\IfNoValueTF{#2}{}{ \mid #2}\right\}}
\DeclareDocumentCommand\transpose{m}{#1^{\intercal}}
\DeclareDocumentCommand\zerovec{o}{\IfNoValueTF{#1}{\mathbb{O}}{\mathbb{O}_{#1}}}
\DeclareDocumentCommand\aff{o}{\operatorname{aff}\IfValueTF{#1}{\left(#1\right)}{}}
\DeclareDocumentCommand\supp{o}{\operatorname{supp}\IfValueTF{#1}{\left(#1\right)}{}}
\DeclareDocumentCommand\orderO{o}{\mathcal{O}\IfValueTF{#1}{\left(#1\right)}{}}
\DeclareDocumentCommand\orderTheta{o}{\Theta\IfValueTF{#1}{\left(#1\right)}}

\DeclareDocumentCommand\reviewComment{mm}{\todo[inline,color=red!50!white]{\textbf{Reviewer:} #1 \textbf{Response:} #2}}

\DeclareDocumentCommand\reviewComment{mm}{}

\declaretheoremstyle[
spaceabove=6pt, spacebelow=6pt,
headfont=\normalfont\itshape,
notefont=\mdseries, notebraces={(}{)},
bodyfont=\normalfont\itshape,
postheadspace=1em
]{myclaim}

\declaretheorem[name=Theorem]{theorem}
\declaretheorem[name=Lemma,sibling=theorem]{lemma}
\declaretheorem[name=Corollary,sibling=theorem]{corollary}
\declaretheorem[name=Proposition,sibling=theorem]{proposition}
\declaretheorem[name=Claim,parent=theorem,style=myclaim]{claim}

\title{A Polyhedral Study for the Cubic Formulation of the Unconstrained Traveling Tournament Problem}

\author{Marije R.\ Siemann}
\author{Matthias Walter}
\affil{University of Twente, The Netherlands}

\date{\small\today}

\begin{document}

\maketitle

\begin{abstract}
  We consider the unconstrained traveling tournament problem, a sports timetabling problem that minimizes traveling of teams.
  Since its introduction about 20 years ago, most research was devoted to modeling and reformulation approaches.
  In this paper we carry out a polyhedral study for the cubic integer programming formulation by establishing the dimension of the integer hull as well as of faces induced by model inequalities.
  Moreover, we introduce a new class of inequalities and show that they are facet-defining.
  Finally, we evaluate the impact of these inequalities on the linear programming bounds.
\end{abstract}

\section{Introduction}
\label{sec_intro}

The \emph{traveling tournament problem} is an optimization problem that involves aspects from tournament timetabling as well as from tour problems such as the traveling salesman problem.
It was introduced by Easton, Nemhauser and Trick in 2001~\cite{EastonNT01}.
To formally state the problem, we consider an even number $n \geq 4$ of sports teams, each playing at its own venue, and the problem of designing a \emph{double round-robin tournament}.
Such a tournament consists of \emph{slots} $S \coloneqq \{1,2,\dotsc,2n-2\}$ and in each slot, each team $i \in V \coloneqq \{1,2,\dotsc,n\}$ plays against another team $j \in V$, either at its home venue or away, i.e., at $j$'s home venue.
Moreover, every two teams $i,j \in V$ play each other exactly twice, once at $i$ and once at $j$.
Finally, distances $d_{i,j}$ between the venues $i,j \in V$ for all pairs $(i,j) \in A \coloneqq \{ (i,j) \in V \times V : i \neq j \}$ are given and the goal is to find a tournament with the minimum total traveling distance.
Between two consecutive slots in which a team plays at different venues $j$ and $k$, it travels $d_{j,k}$ units.
In particular, if both matches are played away, then it directly travels from venue $j$ to venue $k$.
Before slot $1$ and after slot $2n-2$ each team shall reside at its home venue, i.e., if the first or last match is played away, then the team has to travel between this venue and its home venue.
This problem is known as the \emph{unconstrained traveling tournament problem (TTP)}, which is known to be NP-hard~\cite{Bhattacharyya16}.

There exist several variants, including the \emph{classic TTP}.
Here, the unconstrained TTP is further restricted by requiring that the two matches of teams $i$ and $j$ shall not be in consecutive slots.
Moreover, no team shall play more than $3$ consecutive home matches and no more than $3$ consecutive away matches.
This variant is also NP-hard~\cite{ThielenW11}.

The first solution approaches were developed in~\cite{EastonNT03}, where a column generation framework was combined with constraint programming techniques.
The authors of~\cite{MeloUR09} discuss several integer programming formulations in their paper on a single-round-robin variant of the TTP.
In particular, they describe a cubic formulation (with $\orderO(n^3)$ variables) that naturally generalizes to one for the unconstrained TTP.

Already the tournament construction without a traveling aspect is nontrivial.
While there exist several efficient methods to construct a feasible solution (see~\cite{DrexlK07,Froncek10,Kirkman47,RasmussenT08}), the addition of more constraints or an objective function often makes the problem intractable.
For instance, the optimization version, called the \emph{planar 3-index assignment problem}, is NP-hard~\cite{Frieze83}.
However, there exist several polyhedral studies in which the integer hull of the natural integer programming formulation of the planar 3-index assignment problem was investigated~\cite{AppaMM06,EulerBG86,MagosM09}.

\paragraph{Outline.}
In \cref{sec_polytope} we introduce the cubic integer programming formulation in order to define the unconstrained traveling tournament polytope as its integer hull.
In \cref{sec_dimension} we deal with equations valid for the polytope and establish its dimension.
Moreover, in \cref{sec_model_inequalities} we show that some of the model inequalities are facet-defining while others are lifted to have this property.
Finally, in \cref{sec_strengthening_inequalities} we introduce a new class of inequalities and show that they are facet-defining.
For the proofs in \cref{sec_dimension,sec_model_inequalities,sec_strengthening_inequalities} we need to construct tournaments with a variety of properties.
These constructions can be found in \cref{sec_tournaments}.
The paper is concluded in \cref{sec_computations} where we evaluate the impact of our findings computationally.

\section{The unconstrained traveling tournament polytope}
\label{sec_polytope}

\DeclareDocumentCommand\allMatches{}{\mathcal{M}}

A \emph{match} between $i \in V$ and $j \in V \setminus \{i\}$ at venue $i$ that is played in slot $k \in S$ is denoted by the triple $(k,i,j)$ and by $\allMatches$ we denote the set of all possible matches.
The formulation has \emph{play variables} $x_m \in \{0,1\}$ for each match $m \in \allMatches$ and \emph{travel variables} $y_{t,i,j} \in \{0,1\}$ for all $t \in V$ and all $(i,j) \in A$.
The interpretation is that $x_{k,i,j} = 1$ if and only if match $(k,i,j)$ is played, and $y_{t,i,j} = 1$ if (but not only if) team $t$ travels from venue $i$ to venue $j$.
Note that in a tournament each team travels along such an arc at most once.
The formulation reads
\begin{subequations}
  \label{model_basic}
  \begin{alignat}{7}
    & \text{min }
      ~ \mathrlap{ \sum_{(i,j) \in A} d_{i,j} \sum_{t \in V} y_{t,i,j} } \label{model_basic_obj} \\
    & \text{s.t. }
      & \sum_{j \in V \setminus \{i\}} (x_{k,i,j} + x_{k,j,i}) &= 1 &\quad& \forall k \in S : k \geq 2,~ \forall i \in V, \label{model_basic_team_plays} \\
    & & \sum_{k \in S} x_{k,i,j} &= 1 &\quad& \forall (i,j) \in A, \label{model_basic_pair_plays} \\
    & & x_{k,i,t} + x_{k+1,j,t} - 1 &\leq y_{t,i,j} &\quad& \forall k \in S \setminus \{2n-2\}, ~\forall (i,j) \in A, ~\forall t \in V \setminus \{i,j\}, \label{model_basic_travel_away_away} \\
    & & \sum_{i \in V \setminus \{t\}} x_{k,t,i} + x_{k+1,j,t} - 1 &\leq y_{t,t,j} &\quad& \forall k \in S \setminus \{2n-2\}, ~\forall (t,j) \in A, \label{model_basic_travel_home_away} \\
    & & x_{k-1,i,t} + \sum_{j \in V \setminus \{t\}} x_{k,t,j} - 1 &\leq y_{t,i,t} &\quad& \forall k \in S \setminus \{1\}, ~\forall (i,t) \in A, \label{model_basic_travel_away_home}\\
    & & x_{1,j,t} &\leq y_{t,t,j}  &\quad& \forall (t,j) \in A, \label{model_basic_travel_first} \\
    & & x_{2n-2,i,t} &\leq y_{t,i,t} &\quad& \forall (i,t) \in A, \label{model_basic_travel_last} \\
    & & x_{k,i,j} &\in \{0,1\} &\quad& \forall (k,i,j) \in \allMatches, \label{model_basic_domain_x} \\
    & & y_{t,i,j} &\in \{0,1\} &\quad& \forall (t,i,j) \in V \times A. \label{model_basic_domain_y}
  \end{alignat}
\end{subequations}

The objective~\eqref{model_basic_obj} minimizes the total traveled distance.
Constraints~\eqref{model_basic_team_plays} ensure that each team plays exactly once (either home or away) in each slot $k \geq 2$.
For $k=1$, the same equations are implied (see \cref{thm_redundant_equations}).
Constraints~\eqref{model_basic_pair_plays} ensure that each home-away pair occurs exactly once.
This constitutes a correct model for a double round-robin schedule with binary variables $x$.
Note that the \emph{classic} traveling tournament instances also require custom constraints such as a no-repeater constraint (requiring that the two matches of two teams are not scheduled in a row) and upper bounds on the number of consecutive home/away games.
However, for our polyhedral study we omit these constraints to keep the model simple.
The remaining constraints~\eqref{model_basic_travel_away_away}--\eqref{model_basic_travel_last} force the travel variables to be $1$ if the corresponding travel occurs.

\DeclareDocumentCommand\Ptt{}{\ensuremath{P_{\mathrm{utt}}}}

To carry out a polyhedral study, it is worth to define the integer hull of IP~\eqref{model_basic}.
To this end, we define a \emph{tournament} as a subset $T \subseteq \allMatches$ of matches whose \emph{play vector} $\chi(T) \in \{0,1\}^{\allMatches}$, defined via $\chi(T)_{k,i,j} = 1 \iff (k,i,j) \in T$, satisfies~\eqref{model_basic_team_plays} and~\eqref{model_basic_pair_plays}.
Its \emph{travel vector} is the vector $\psi(T) \in \{0,1\}^{V \times A}$ with $\psi(T)_{t,i,j} = 1$ if and only if team $t$ travels from venue $i$ to venue $j$.
In the IP, a travel variable $y_{t,i,j}$ can be set to $1$ although team $t$ does not travel from $i$ to $j$.
If the distances $d_{i,j}$ are positive, this will however never happen in an optimal solution.
The integer hull of the IP, which we call the \emph{unconstrained traveling tournament polytope}, is thus equal to
\begin{equation*}
  \Ptt(n) \coloneqq \conv\{ (\chi(T),y) \in \{0,1\}^\allMatches \times \{0,1\}^{V \times A} : T \text{ tournament and } y \geq \psi(T) \}.
\end{equation*}
A natural question is why we require $y \geq \psi(T)$ in the definition.
This is accordance with the existing integer programming model the literature (see~\cite{MeloUR09}), in which $y$-variables are only constrained from below.
The variant in which $y = \psi(T)$ is enforced is much harder to study theoretically, but we will later obtain the corresponding polytope as a face of $\Ptt(n)$, see \cref{thm_face}.
Finally, by $\zerovec$ we denote the zero vector, where its length can be derived from the context.
Note that $\Ptt(n)$ resembles the basic properties of the traveling tournament model.
More restrictions and corresponding linear constraints are discussed in Section~\ref{sec_problem_variants}.

\section{Equations and dimension}
\label{sec_dimension}

\subsection{Known equations}

\begin{proposition}
  \label{thm_redundant_equations}
  For each team $t \in V$, equations~\eqref{model_basic_team_plays} for $(k,i) = (1,t)$ follow from equations~\eqref{model_basic_team_plays} for all $k \in S \setminus \{1\}$ and $i = t$ together with equations~\eqref{model_basic_pair_plays} for all $(i,j) \in A$ with $t \in \{i,j\}$.
\end{proposition}

\begin{proof}
  Let $t \in V$.
  The sum of equations~\eqref{model_basic_pair_plays} for all $(i,j) \in A$ with $j = t$ plus the sum of equations~\eqref{model_basic_pair_plays} for all $(i,j) \in A$ with $i = t$ minus the sum of equations~\eqref{model_basic_team_plays} for all $k \in S \setminus \{1\}$ and $i=t$ yields
  \begin{multline*}
    \sum_{i \in V \setminus \{t\}} \sum_{k \in S} x_{k,i,t} 
    + \sum_{j \in V \setminus \{t\}} \sum_{k \in S} x_{k,t,j}
    - \sum_{k \in S \setminus \{1\}} \sum_{j \in V \setminus \{t\}} (x_{k,t,j} + x_{k,j,t}) \\
    = (n-1) + (n-1) - (2n-3)
    \iff \sum_{j \in V \setminus \{t\}} (x_{1,t,j} + x_{1,j,t})
    = 1,
  \end{multline*}
  which is equation~\eqref{model_basic_team_plays} for $(k,i) = (1,t)$.
\end{proof}

\DeclareDocumentCommand\columnBasis{}{\mathcal{B}}

We define the following \emph{column basis} $\columnBasis_{\bar{k}} \subseteq \allMatches$ via
\begin{equation}
  \columnBasis_{\bar{k}} \coloneqq \{ (k,i,j) \in \allMatches : k = \bar{k} \text{ or } i = 1 \text{ or } (i,j) = (2,3) \}. \label{eq_basis}
\end{equation}
We will often use the following lemma which states that the play variables indexed by $\columnBasis_{\bar{k}}$ induce an invertible submatrix of the equation system of interest.

\begin{lemma}
  \label{thm_dim_basis}
  Let $\bar{k} \in S$ and let $Cx = d$ be the system defined by equations~\eqref{model_basic_team_plays} and~\eqref{model_basic_pair_plays}.
  Then the submatrix of $C$ induced by variables $x_m$ for $m \in \columnBasis_{\bar{k}}$ is invertible.
  In particular, these $|\columnBasis_{\bar{k}}| = 3n^2-4n$ equations are irredundant.
\end{lemma}

\begin{proof}
  Observe that variables $x_{\bar{k},i,j}$ only appear in equation~\eqref{model_basic_pair_plays} for $(i,j) \in A$.
  Thus, by cofactor expansion it remains to prove invertibility of the coefficient submatrix $C'$ of $C$ whose rows correspond to equations~\eqref{model_basic_team_plays} and whose columns correspond to variables $x_{k,i,j}$ for $(k,i,j) \in \allMatches$ with $k \neq \bar{k}$ and $i=1$ or $(i,j) = (2,3)$.

  The matrix $C'$ is a block diagonal matrix. The blocks are the submatrices $C^k$ whose rows and columns are the same as those of $C'$ but for fixed $k$.
  For the remainder of the proof we fix $k \in S \setminus \{\bar{k}\}$ and prove that $C^k$ is invertible.
  For $\ell \in \{3,4,\dotsc,n\}$, consider the submatrices $C^{k,\ell} \in \R^{\ell \times \ell}$ of $C$ induced by equations~\eqref{model_basic_team_plays} for $k$ and for $i=1,2,\dotsc,\ell$ and by variables $x_{k,2,3}$, $x_{k,1,2}$, $x_{k,1,3}$, $x_{k,1,4}$, \dots, $x_{k,1,\ell}$.
  One easily verifies that $C^{k,3}$ is invertible and that for $\ell \geq 4$, $C^{k,\ell}$ is obtained from $C^{k,\ell-1}$ by adding a unit row with the one in the added column.
  By induction on $\ell$, cofactor expansion shows that $C^{k,n}$ is invertible.
  The fact that $C^{k,n} = C^k$ holds, concludes the proof.
\end{proof}

A consequence of \cref{thm_dim_basis} is that every equation that is valid for $\Ptt(n)$ or some of its faces can be turned into an equivalent one that involves no $x_m$ for $m \in \columnBasis_{\bar{k}}$.
Hence, in many subsequent proofs we will assume that such an equation $\transpose{a}x + \transpose{b}y = \gamma$ satisfies
\begin{equation}
  a_m = 0 \text{ for each } m \in \columnBasis \coloneqq \columnBasis_1,
  \tag{$\columnBasis$}
  \label{eq_basis_zero}
\end{equation}
and refer to this as requiring \emph{the equation to be normalized with respect to slot $\bar{k} = 1$.} in our proofs.

\subsection{Tournaments from 1-factors}

We consider the tournament construction based on perfect matchings (also called $1$-factors) of the complete graphs on $n$ nodes (see~\cite{DrexlK07}).
In each tournament $T$, for each $k \in S$, the matches $(k,i,j) \in T$ in slot $k$, interpreted as edges $\{i,j\}$, form a perfect matching.
Thus, each tournament is characterized by $|S|$ such perfect matchings whose edges are oriented so that no arc $(i,j) \in A$ appears twice.
Since the latter is the only restriction, we can first determine the $|S|$ perfect matchings $M_k$ for all $k\in S$ and afterwards orient their edges in a complementary fashion, that is,
\begin{equation}
  \text{each edge $\{i,j\}$ is oriented differently in the two  perfect matchings in which it is contained.}
  \label{eq_canonical_factorization_orientation}
\end{equation}
We call such an orientation \emph{complementary}.
The following \emph{canonical factorization} is one specific set $\{M_1, M_2, \dotsc, M_{2n-2}\}$ of perfect matchings~\cite{DrexlK07}, where $M_k$ for $k < n$ is determined by 
\begin{equation*}
  M_k \coloneqq \{ \{ k, n \} \} \cup \{ \{ k+i, k-i \} : i=1,2,\dotsc,n/2-1 \}, 
\end{equation*}
where $k+i$ and $k-i$ are taken modulo $n-1$ as one of the numbers $1,2,\dotsc,n-1$.
The remaining perfect matchings are $M_k \coloneqq M_{k - n + 1}$ for all $k \in \{n,n+1,\dotsc,2n-2\}$.
Hence,
\begin{multline}
  \text{for each edge $\{i,j\}$ there is a unique $k \in \{1,2,\dotsc,n-1\}$ with $\{i,j\} \in M_k$ and $\{i,j\} \in M_{k+n-1}$} \\ \text{a unique $k' \in \{n,n+1,\dotsc,2n-2\}$ with $\{i,j\} \in M_{k'}$ (which satisfies $k' = k + n -1$).}
  \label{eq_canonical_factorization_property}
\end{multline}
We will often construct tournaments obtained from the canonical factorizations by permuting slots or teams.
In many cases, it is easy to see that corresponding permutations exist.
Hence, we typically state that a tournament is constructed from a canonical factorization such that certain requirements are satisfied, e.g., by specifying certain matches that shall be played.

\paragraph{Operations on tournaments.}
Three fundamental operations to modify a given tournament are the \emph{cyclic shift}, the \emph{home-away swap} and the \emph{partial slot swap}, defined as follows.

Let $s \in \Z$ and let $T$ be a tournament.
We say that tournament $T'$ is obtained by a \emph{cyclic shift by $s$} if $T'$ arises from $T$ by mapping each slot $k \in S$ to slot $k + s$, where slots are considered modulo $2n-2$ in the range $1,2,\dotsc,2n-2$.

\begin{proposition}[Home-away swap]
  \label{thm_home_away_swap}
  Let $T$ be a tournament with matches $(k_1,i,j), (k_2,j,i) \in T$.
  Then
  \begin{equation}
    T' \coloneqq T \setminus \{ (k_1,i,j),(k_2,j,i) \} \cup \{ (k_1,j,i), (k_2,i,j) \}
    \tag{HA$_{k_1,k_2,i,j}$} \label{eq_home_away_swap}
  \end{equation}
  is also a tournament.
\end{proposition}

\DeclareDocumentCommand\refHomeAwaySwap{mmmm}{\hyperref[eq_home_away_swap]{\ensuremath{\text{HA}_{\text{#1},\text{#2},\text{#3},\text{#4}}}}}
\DeclareDocumentCommand\eqrefHomeAwaySwap{mmmm}{(\refHomeAwaySwap{#1}{#2}{#3}{#4})}
\DeclareDocumentCommand\refHomeAwaySwapDefault{}{\refHomeAwaySwap{$1$}{$k$}{$i$}{$j$}}
\DeclareDocumentCommand\eqrefHomeAwaySwapDefault{}{\eqrefHomeAwaySwap{$1$}{$k$}{$i$}{$j$}}

\begin{proposition}[Partial slot swap]
  \label{thm_partial_slot_swap}
  Let $T$ be a tournament with matches $(k_1,i,j)$, $(k_1,i',j')$, $(k_2,i,j')$, $(k_2,i',j) \in T$.
  Then
  \begin{multline}
    \qquad T' \coloneqq T \setminus \{ (k_1,i,j), (k_1,i',j'), (k_2,i,j'), (k_2,i',j) \} \\ 
      \cup \{ (k_1,i,j'), (k_1,i',j), (k_2,i,j), (k_2,i',j') \} \qquad \tag{PS$_{k_1,k_2,i,j,i',j'}$} \label{eq_partial_slot_swap}
  \end{multline}
  is also a tournament.
\end{proposition}

\DeclareDocumentCommand\refPartialSlotSwap{mmmmmm}{\hyperref[eq_partial_slot_swap]{\ensuremath{\text{PS}_{\text{#1},\text{#2},\text{#3},\text{#4},\text{#5},\text{#6}}}}}
\DeclareDocumentCommand\eqrefPartialSlotSwap{mmmmmm}{(\refPartialSlotSwap{#1}{#2}{#3}{#4}{#5}{#6})}
\DeclareDocumentCommand\refPartialSlotSwapDefault{}{\refPartialSlotSwap{$1$}{$k$}{$i$}{$j$}{$i'$}{$j'$}}
\DeclareDocumentCommand\eqrefPartialSlotSwapDefault{}{(\refPartialSlotSwapDefault)}

\subsection{Dimension of the unconstrained traveling tournament polytope}

\begin{theorem}
  \label{thm_dim}
  The affine hull of $\Ptt(n)$ is described completely by the irredundant equations~\eqref{model_basic_team_plays} and~\eqref{model_basic_pair_plays}.
\end{theorem}

\begin{proof}
  We first observe that \cref{thm_dim_basis} implies that the equations are irredundant, i.e., none of them is a linear combination of the others.
  It remains to prove that every valid equation is a linear combination of these equations.
  To this end, we show that for any equation $\transpose{a}x + \transpose{b}y = \gamma$ that is valid for $\Ptt(n)$ and that is normalized with respect to slot $1$ (i.e., it satisfies~\eqref{eq_basis_zero}) that $(a,b) = \zerovec$ holds.

  \begin{restatable}{claim}{thmDimNoTravel}
    \label{thm_dim_no_travel}
    For each $(t,i,j) \in V \times A$ there exists a tournament in which team $t$ never travels from venue $i$ to venue $j$.
  \end{restatable}

  A tournament $T$ from \cref{thm_dim_no_travel} satisfies $\psi(T)_{t,i,j} = 0$.
  Let $y \coloneqq \psi(T)$ and let $y'$ be equal to $y$ except for $y'_{t,i,j} = 1$.
  Hence, $(\chi(T),y), (\chi(T),y') \in \Ptt(n)$ and thus $\transpose{a}\chi(T) + \transpose{b}y = \gamma = \transpose{a}\chi(T) + \transpose{b}y'$ holds.
  We obtain
  \begin{equation}
    b = \zerovec.
    \tag{\S\ref*{thm_dim_no_travel}}
    \label{proof_dim_no_travel}
  \end{equation}

  \begin{restatable}{claim}{thmDimHomeAway}
    \label{thm_dim_home_away}
    For each $k \in S \setminus \{1\}$ and for distinct $i,j \in V$ there exist tournaments $T$ and $T'$ satisfying~\eqrefHomeAwaySwapDefault{}.
  \end{restatable}

  For the tournaments $T$ and $T'$ from \cref{thm_dim_home_away} we have $\transpose{b}\psi(T) = \transpose{b}\psi(T')$ due to~\eqref{proof_dim_no_travel}.
  Using the fact that the equation is normalized with respect to slot $1$ (i.e., it satisfies~\eqref{eq_basis_zero}), $\transpose{a}\chi(T) + \transpose{b}\psi(T) = \gamma = \transpose{a}\chi(T') + \transpose{b}\psi(T')$ simplifies to
  \begin{equation}
    a_{k,i,j} = a_{k,j,i} \text{ for each } k \in S \setminus \{1\} \text{ and for all distinct } i,j \in V.
    \tag{\S\ref*{thm_dim_home_away}}
    \label{proof_dim_home_away}
  \end{equation}

  \begin{restatable}{claim}{thmDimPartialSlot}
    \label{thm_dim_partial_slot}
    For each $k \in S \setminus \{1\}$ and for distinct $i,j,i',j' \in V$ there exist tournaments $T$ and $T'$ satisfying~\eqrefPartialSlotSwapDefault{}.
  \end{restatable}

  For the tournaments $T$ and $T'$ from \cref{thm_dim_partial_slot} we have $\transpose{b}\psi(T) = \transpose{b}\psi(T')$ due to~\eqref{proof_dim_no_travel}.
  Using \eqref{eq_basis_zero}, $\transpose{a}\chi(T) + \transpose{b}\psi(T) = \gamma = \transpose{a}\chi(T') + \transpose{b}\psi(T')$ simplifies to
  \begin{equation}
    a_{k,i,j} + a_{k,i',j'} = a_{k,i,j'} + a_{k,i',j} \text{ for each } k \in S \setminus \{1\} \text{ and for all distinct } i,j,i',j' \in V.
    \tag{\S\ref*{thm_dim_partial_slot}}
    \label{proof_dim_partial_slot}
  \end{equation}

  Consider a slot $k \in S \setminus \{1\}$.
  For each $\ell \in \{4,5,\dotsc,n\}$, \eqref{proof_dim_partial_slot} implies $a_{k,1,\ell} + a_{k,2,3} = a_{k,1,3} + a_{k,2,\ell}$ which, together
  with the fact that the equation is normalized with respect to slot $1$ (i.e., it satisfies~\eqref{eq_basis_zero}), yields $a_{k,2,\ell} = 0$.
  Combined with~\eqref{proof_dim_home_away} we also obtain $a_{k,\ell,2} = 0$.
  For all distinct $\ell,\ell' \in \{3,4,\dotsc,n\}$, \eqref{proof_dim_partial_slot} implies $a_{k,1,\ell'} + a_{k,\ell,2} = a_{k,1,2} + a_{k,\ell,\ell'}$.
  Together with \eqref{eq_basis_zero}, this shows $a_{k,\ell,\ell'} = 0$.
  Hence, $a = \zerovec$ holds, which concludes the proof.
\end{proof}

\begin{corollary}
  The dimension of $\Ptt(n)$ is equal to $3n^3 - 8n^2 + 6n$.
\end{corollary}

\begin{proof}
  The ambient space of $\Ptt(n)$ has dimension $|\allMatches| + n \cdot |A|$.
  By \cref{thm_dim}, the affine hull is described by the $3n^2 - 4n$ equations~\eqref{model_basic_team_plays} and~\eqref{model_basic_pair_plays}, which are irredundant by \cref{thm_dim_basis}.
  Hence,
  \begin{equation*}
    \dim(\Ptt(n)) = (2n-2) \cdot n \cdot (n-1) + n \cdot n \cdot (n-1) - (3n^2 - 4n) = 3n^3 - 8n^2 + 6n. 
  \end{equation*}
  This concludes the proof.
\end{proof}

\section{Model inequalities}
\label{sec_model_inequalities}

In this section we consider the inequalities from~\eqref{model_basic} and determine when they are facet-defining.
While for many integer programming problems such a verification is a simple task that does not yield any insight, in our case we already observe that establishing facetness is nontrivial.
This is due to the combinatorics of tournament schedules which do not admit the construction of simple (and affinely independent) solution vectors.
This complexity is already indicated by \cref{thm_partial_slot_swap} where eight coordinates must be changed in a very structured way in order to move from one solution vector to another.
Moreover, our attempts to prove that inequalites~\eqref{model_basic_travel_away_away} are facet-defining failed, and it turned out that they actually are not.
However, they are almost facet-defining in the sense that the dimension of their induced face is to low by $1$, and we provide the corresponding two facets that have this face as their intersection.
Within the proofs we will sometimes argue about symmetry of the formulation, for which we state the following lemma.

\begin{lemma}
  \label{thm_symmetry}
  $\Ptt(n)$ and formulation~\eqref{model_basic} are symmetric with respect to permuting teams and with respect to mirroring all slots, i.e., exchanging roles of slots $k$ and $2n-1-k$ for all $k \in \{1,2,\dotsc,n-1\}$.
\end{lemma}

\begin{proof}
  Symmetry with respect to team permutations is clear for $\Ptt(n)$ and for the formulation.

  Moreover, symmetry with respect to mirroring slots is easy to see for $\Ptt(n)$: when slots are exchanged, all traveled arcs are simply reversed.
  For the formulation, the roles of~\eqref{model_basic_travel_home_away} and~\eqref{model_basic_travel_away_home} as well as~\eqref{model_basic_travel_first} and~\eqref{model_basic_travel_last} are exchanged.
\end{proof}

We start with the nonnegativity constraints for the play variables.
\begin{theorem}
  \label{thm_nonneg}
  Inequalities $x_{k,i,j} \geq 0$ are facet-defining for $\Ptt(n)$ for all $(k,i,j) \in \allMatches$.
\end{theorem}

\begin{proof}
  Consider the inequality $x_{k^\star,i^\star,j^\star} \geq 0$ for some match $m^\star = (k^\star,i^\star,j^\star) \in \allMatches$.
  By \cref{thm_symmetry}, we can assume $k^\star \geq n$ and $i^\star = 3$ and $j^\star = 4$.
  This implies $m^\star \notin \columnBasis$.
  Let $\transpose{a}x + \transpose{b}y \geq \gamma$ define any facet $F$ that contains the face induced by this inequality.
  Without loss of generality, the equation is normalized with respect to slot $1$, i.e., it satisfies~\eqref{eq_basis_zero}.
  It remains to prove that $b = \zerovec$ and $\gamma = 0$ hold and that $a$ is a multiple of $\chi(\{(k^\star,i^\star,j^\star)\})$.

  \begin{restatable}{claim}{thmNonnegNoTravel}
    \label{thm_nonneg_no_travel}
    For each $(t,i,j) \in V \times A$ there exists a tournament $T$ with $m^\star \notin T$ and in which team $t$ never travels from venue $i$ to venue $j$.
  \end{restatable}

  A tournament $T$ from \cref{thm_nonneg_no_travel} satisfies $\psi(T)_{t,i,j} = 0$.
  Let $y \coloneqq \psi(T)$ and let $y'$ be equal to $y$ except for $y'_{t,i,j} = 1$.
  Since $\chi(T)_{m^\star} = 0$ holds, we have $(\chi(T),y), (\chi(T),y') \in F$.
  The equation $\transpose{a}\chi(T) + \transpose{b}y = \gamma = \transpose{a}\chi(T) + \transpose{b}y'$ simplifies to $b_{t,i,j} = 0$.
  We obtain
  \begin{equation}
    b = \zerovec.
    \tag{\S\ref*{thm_nonneg_no_travel}}
    \label{proof_nonneg_no_travel}
  \end{equation}

  \begin{restatable}{claim}{thmNonnegHomeAway}
    \label{thm_nonneg_home_away}
    For each $(k,i,j) \in \allMatches$ with $k \geq 2$ and $(k,i,j) \neq (k^\star,i^\star,j^\star),(k^\star,j^\star,i^\star)$ there exist tournaments $T$ and $T'$ satisfying~\eqrefHomeAwaySwapDefault{} and $(k^\star,i^\star,j^\star),(k^\star,j^\star,i^\star) \notin T \cup T'$.
  \end{restatable}

  The tournaments $T$ and $T'$ from \cref{thm_nonneg_home_away} satisfy $\chi(T)_{m^\star} = \chi(T')_{m^\star} = 0$ and thus we have $(\chi(T),\psi(T)), (\chi(T'),\psi(T')) \in F$.
  Using the fact that the equation is normalized with respect to slot $1$ (i.e., it satisfies~\eqref{eq_basis_zero}) and~\eqref{thm_nonneg_no_travel}, $\transpose{a}\chi(T) + \transpose{b}\psi(T) = \gamma = \transpose{a}\chi(T') + \transpose{b}\psi(T')$ simplifies to
  \begin{equation}
    a_{k,i,j} = a_{k,j,i} \text{ for each } (k,i,j) \in \allMatches \text{ with } k \geq 2 \text{ and } (k,i,j) \notin \{ (k^\star,i^\star,j^\star), (k^\star,j^\star,i^\star) \}.
    \tag{\S\ref*{thm_nonneg_home_away}}
    \label{proof_nonneg_home_away}
  \end{equation}

  \begin{restatable}{claim}{thmNonnegPartialSlot}
    \label{thm_nonneg_partial_slot}
    For each slot $k \in S \setminus \{1\}$ and for distinct $i,j,i',j' \in V$ with $k \neq k^\star$ or $(i^\star,j^\star) \notin \{ (i,j),(i',j'),(i',j),(i,j') \}$ there exist tournaments $T$ and $T'$ satisfying~\eqrefPartialSlotSwapDefault{} and $m^\star \notin T \cup T'$.
  \end{restatable}

  The tournaments $T$ and $T'$ from \cref{thm_nonneg_partial_slot} satisfy $\chi(T)_{m^\star} = \chi(T')_{m^\star} = 0$ and thus we have $(\chi(T),\psi(T)), (\chi(T'),\psi(T')) \in F$.
  Using the fact that the equation is normalized with respect to slot $1$ (i.e., it satisfies~\eqref{eq_basis_zero})  and~\eqref{thm_nonneg_no_travel}, $\transpose{a}\chi(T) + \transpose{b}\psi(T) = \gamma = \transpose{a}\chi(T') + \transpose{b}\psi(T')$ simplifies to
  \begin{multline}
    a_{k,i,j} + a_{k,i',j'} = a_{k,i,j'} + a_{k,i',j} \text{ for each } k \in S \setminus \{1\} \text{ and for all distinct } (i,j,i',j') \in V \\ \text{ with } k \neq k^\star \text{ or } (i^\star,j^\star) \notin \{ (i,j),(i',j'),(i,j'),(i',j) \}.
    \tag{\S\ref*{thm_nonneg_partial_slot}}
    \label{proof_nonneg_partial_slot}
  \end{multline}

  Consider a slot $k \in S \setminus \{1\}$.
  For each $\ell \in \{4,5,\dotsc,n\}$, \eqref{proof_nonneg_partial_slot} implies $a_{k,1,\ell} + a_{k,2,3} = a_{k,1,3} + a_{k,2,\ell}$ which, together with the fact that the equation is normalized with respect to slot $1$ (i.e., it satisfies~\eqref{eq_basis_zero}), yields $a_{k,2,\ell} = 0$.
  Combined with~\eqref{proof_nonneg_home_away} we also obtain $a_{k,\ell,2} = 0$.
  For all distinct $\ell,\ell' \in \{3,4,\dotsc,n\}$ except for $(\ell,\ell') = (4,3)$, \eqref{proof_nonneg_partial_slot} implies $a_{k,1,\ell'} + a_{k,\ell,2} = a_{k,1,2} + a_{k,\ell,\ell'}$.
  Together with the fact that the equation is normalized with respect to slot $1$ (i.e., it satisfies~\eqref{eq_basis_zero}), this shows $a_{k,\ell,\ell'} = 0$ for all but the entry corresponding to match $(k^\star,i^\star,j^\star)$.

  Hence, the inequality reads $a_{k^\star,i^\star,j^\star} \cdot x_{k^\star,i^\star,j^\star} \geq \gamma$.
  Since $\chi(T)_{k^\star,i^\star,j^\star} = 0$ holds for each of the considered tournaments $T$, we obtain $\gamma = 0$.
  Finally, since there exist tournaments $T$ for which $\chi(T)_{k^\star,i^\star,j^\star} = 1$ holds, $a_{k^\star,i^\star,j^\star}$ must be positive, which concludes the proof.
\end{proof}

We continue with inequalities~\eqref{model_basic_travel_away_away} which are not facet-defining.
However, they can be lifted to these two stronger ones.
\begin{subequations}
  \label{eq_travel_away_away_lifted}
  \begin{align}
    x_{k,j,t} + x_{k,i,t} + x_{k+1,j,t} - 1 &\leq y_{t,i,j} \quad \forall  k \in S \setminus \{2n-2\}, ~\forall (i,j) \in A, ~\forall t \in V \setminus \{i,j\} \label{eq_travel_away_away_lifted_first} \\
    x_{k+1,i,t} + x_{k,i,t} + x_{k+1,j,t} - 1 &\leq y_{t,i,j} \quad \forall  k \in S \setminus \{2n-2\}, ~\forall (i,j) \in A, ~\forall t \in V \setminus \{i,j\} \label{eq_travel_away_away_lifted_second}
  \end{align}
\end{subequations}
Indeed, in order to obtain~\eqref{model_basic_travel_away_away} they only need to be combined with nonnegativity constraints for $x$.
These inequalities turn out to be facet-defining.

\begin{theorem}
  \label{thm_travel_away_away_lifted}
  Inequalities~\eqref{eq_travel_away_away_lifted} are facet-defining for $\Ptt(n)$ for each slot $k \in S \setminus \{2n-2\}$ and all distinct teams $i,j,t \in V$.
\end{theorem}

\begin{proof}
  We only prove the statement for inequalities~\eqref{eq_travel_away_away_lifted_first} since the proof for~\eqref{eq_travel_away_away_lifted_second} is similar.
  Moreover, we assume $n \geq 6$ since we verified the statement for $n = 4$ computationally.
  For this, we used the software package IPO~\cite{IPO}, which can exactly compute dimensions of polyhedra that are defined implicitly via an optimization oracle, in this case an MIP solver (see Chapter~2 in~\cite{Walter16} for the algorithmic background).

  Consider the inequality $x_{k^\star,j^\star,t^\star} + x_{k^\star,i^\star,t^\star} + x_{k^\star+1,j^\star,t^\star} - y_{t^\star,i^\star,j^\star} \leq 1$ for some slot $k^\star \in S \setminus \{2n-2\}$, and distinct teams $i^\star,j^\star,t^\star \in V$.
  By \cref{thm_symmetry}, we can assume $k^\star \geq n$, $i^\star = 4$, $j^\star = 5$ and $t^\star = 6$.
  The inequality is valid for $\Ptt(n)$ since the only possibility of scheduling more than one of the three matches $(k^\star,j^\star,t^\star)$, $(k^\star,i^\star,t^\star)$ and $(k^\star+1,j^\star,t^\star)$ consists of the latter two which implies that team $t^\star$ travels from venue $i^\star$ to venue $j^\star$.
  The following claim is used several times throughout the proof.
  \begin{restatable}{claim}{thmTravelAwayAwayLiftedFace}
    \label{thm_travel_away_away_lifted_face}
    Let $T$ be a tournament that contains
    \begin{enumerate}[label={(\alph*)}]
    \item
      \label{thm_travel_away_away_lifted_face_yes}
      match $(k^\star,i^\star,t^\star)$ and in which team $t^\star$ plays away in slot $k^\star+1$, or
    \item
      \label{thm_travel_away_away_lifted_face_no}
      one of the matches $(k^\star,j^\star,t^\star)$, $(k^\star,i^\star,t^\star)$ or $(k^\star+1,j^\star,t^\star)$, and in which team~$t^\star$ never travels from venue~$i^\star$ to venue~$j^\star$.
    \end{enumerate}
    Then $(\chi(T),\psi(T))$ satisfies~\eqref{eq_travel_away_away_lifted_first} with equality.
  \end{restatable}

  In order to prove that the inequality is facet-defining, let $\transpose{a}x + \transpose{b}y \leq \gamma$ define any facet $F$ that contains the face induced by this inequality.
  We will prove that it is a multiple of inequality~\eqref{eq_travel_away_away_lifted_first}.
  Without loss of generality, we assume that the equation is normalized with respect to slot $1$, i.e., it satisfies~\eqref{eq_basis_zero}.

  \begin{restatable}{claim}{thmTravelAwayAwayLiftedNoTravel}
    \label{thm_travel_away_away_lifted_no_travel}
    For all $(t,i,j) \in V \times A$ with $(t,i,j) \neq (t^\star,i^\star,j^\star)$ there exists a tournament $T$ in which team $t$ never travels from venue $i$ to venue $j$ and which satisfies condition~\ref{thm_travel_away_away_lifted_face_yes} of \cref{thm_travel_away_away_lifted_face}.
  \end{restatable}

  A tournament $T$ from \cref{thm_travel_away_away_lifted_no_travel} satisfies $\psi(T)_{t,i,j} = 0$.
  Let $y \coloneqq \psi(T)$ and let $y'$ be equal to $y$ except for $y'_{t,i,j} = 1$.
  By  \cref{thm_travel_away_away_lifted_face} we have $(\chi(T),y), (\chi(T),y') \in F$.
  In this case, $\transpose{a}\chi(T) + \transpose{b}y = \gamma = \transpose{a}\chi(T) + \transpose{b}y'$ simplifies to
  \begin{equation}
    b_{t,i,j} = 0 \text{ for all } (t,i,j) \in V \times A \text{ with } (t,i,j) \neq (t^\star,i^\star,j^\star).
    \tag{\S\ref*{thm_travel_away_away_lifted_no_travel}}
    \label{proof_travel_away_away_lifted_no_travel}
  \end{equation}

  \begin{restatable}{claim}{thmTravelAwayAwayLiftedHomeAway}
    \label{thm_travel_away_away_lifted_home_away}
    For each $(k,i,j) \in \allMatches \setminus \{ (k^\star,i^\star,t^\star)$, $(k^\star,t^\star,i^\star)$, $(k^\star,j^\star,t^\star)$, $(k^\star,t^\star,j^\star)$, $(k^\star+1,j^\star,t^\star)$, $(k^\star+1,t^\star,j^\star) \}$ with $k \geq 2$ there exist tournaments $T$ and $T'$ satisfying~\eqrefHomeAwaySwapDefault{} and condition~\ref{thm_travel_away_away_lifted_face_no} of \cref{thm_travel_away_away_lifted_face}.
  \end{restatable}

  The tournaments $T$ and $T'$ from \cref{thm_travel_away_away_lifted_home_away} satisfy $(\chi(T),\psi(T)), (\chi(T'),\psi(T')) \in F$ by \cref{thm_travel_away_away_lifted_face}.
  Using the fact that the equation is normalized with respect to slot $1$ (i.e., it satisfies~\eqref{eq_basis_zero}) and~\eqref{thm_travel_away_away_lifted_no_travel}, $\transpose{a}\chi(T) + \transpose{b}\psi(T) = \gamma = \transpose{a}\chi(T') + \transpose{b}\psi(T')$ simplifies to
  \begin{multline}
    a_{k,i,j} = a_{k,j,i} \text{ for each } (k,i,j) \in \allMatches \setminus \{ (k^\star,i^\star,t^\star), (k^\star,t^\star,i^\star), \\ (k^\star,j^\star,t^\star), (k^\star,t^\star,j^\star), (k^\star+1,j^\star,t^\star), (k^\star+1,t^\star,j^\star) \}.
    \tag{\S\ref*{thm_travel_away_away_lifted_home_away}}
    \label{proof_travel_away_away_lifted_home_away}
  \end{multline}
  
  \begin{restatable}{claim}{thmTravelAwayAwayLiftedPartialSlot}
    \label{thm_travel_away_away_lifted_partial_slot}
    Let $k \in S \setminus \{1\}$, let $i,j,i',j' \in V$ be distinct and let $P \coloneqq \{ (i,j), (i',j'), (i,j'), (i',j) \}$. If
    \begin{enumerate}[label={(\roman*)}]
    \item
      \label{thm_travel_away_away_lifted_partial_slot_out_out}
      $(i^\star,t^\star) \notin P$ and $(j^\star,t^\star) \notin P$, or
    \item
      \label{thm_travel_away_away_lifted_partial_slot_out_in}
      $(i^\star,t^\star) \notin P$, $(j^\star,t^\star) \in P$ and $k \notin \{k^\star, k^\star+1\}$, or
    \item
      \label{thm_travel_away_away_lifted_partial_slot_in_out}
      $(i^\star,t^\star) \in P$, $(j^\star,t^\star) \notin P$ and $k \neq k^\star$, or
    \item
      \label{thm_travel_away_away_lifted_partial_slot_in_in}
      $(i^\star,t^\star) \in P$, $(j^\star,t^\star) \in P$ and $k = k^\star$
    \end{enumerate}
    holds, then there exist tournaments $T$ and $T'$ satisfying~\eqrefPartialSlotSwapDefault{} and condition~\ref{thm_travel_away_away_lifted_face_no} of \cref{thm_travel_away_away_lifted_face}.
  \end{restatable}
  
  The tournaments $T$ and $T'$ from \cref{thm_travel_away_away_lifted_partial_slot} satisfy $(\chi(T),\psi(T)), (\chi(T'),\psi(T')) \in F$ by \cref{thm_travel_away_away_lifted_face}.
  Using the fact that the equation is normalized with respect to slot $1$ (i.e., it satisfies~\eqref{eq_basis_zero}) and~\eqref{thm_travel_away_away_lifted_no_travel}, $\transpose{a}\chi(T) + \transpose{b}\psi(T) = \gamma = \transpose{a}\chi(T') + \transpose{b}\psi(T')$ simplifies to
  \begin{equation}
    a_{k,i,j} + a_{k,i',j'} = a_{k,i,j'} + a_{k,i',j} \text{ for all } (k,i,j,i',j') \text{ satisfying the conditions in \cref{thm_travel_away_away_lifted_partial_slot}} .
    \tag{\S\ref*{thm_travel_away_away_lifted_partial_slot}a}
    \label{proof_travel_away_away_lifted_partial_slot_a}
  \end{equation}

  Consider a slot $k \in S \setminus \{1\}$.
  For each $\ell \in \{4,5,\dotsc,n\}$, \eqref{proof_travel_away_away_lifted_partial_slot_a} for $(i,j,i',j') = (1,\ell,2,3)$ is applicable since condition~\ref{thm_travel_away_away_lifted_partial_slot_out_out} of \cref{thm_travel_away_away_lifted_partial_slot} is satisfied due to $\{i^\star,j^\star\} \cap \{1,2\} = \varnothing$.
  This implies $a_{k,1,\ell} + a_{k,2,3} = a_{k,1,3} + a_{k,2,\ell}$ which, together with the fact that the equation is normalized with respect to slot $1$ (i.e., it satisfies~\eqref{eq_basis_zero}), yields $a_{k,2,\ell} = 0$.
  Moreover, for each $\ell \in \{3,4,\dotsc,n\}$, \eqref{proof_travel_away_away_lifted_home_away} for $(i,j) = (\ell,2)$ implies $a_{k,\ell,2} = a_{k,2,\ell} = 0$.

  For distinct $\ell, \ell' \in \{3,4,\dotsc,n\}$ with $(k,\ell,\ell') \notin \{ (k^\star, t^\star, i^\star), (k^\star, t^\star, j^\star), (k^\star+1, t^\star, j^\star) \}$, \eqref{proof_travel_away_away_lifted_partial_slot_a} for $(i,j,i',j') = (1,\ell',\ell,2)$ is applicable, which implies $a_{k,1,\ell'} + a_{k,\ell,2} = a_{k,1,2} + a_{k,\ell,\ell'}$.
  Together with the fact that the equation is normalized with respect to slot $1$ (i.e., it satisfies~\eqref{eq_basis_zero}) this shows
  \begin{equation}
    a_{k,i,j} = 0 \text{ for all } (k,i,j) \in \allMatches \setminus \{ (k^\star,i^\star,t^\star), (k^\star,j^\star,t^\star), (k^\star+1,j^\star,t^\star) \}.
    \tag{\S\ref*{thm_travel_away_away_lifted_partial_slot}b}
    \label{proof_travel_away_away_lifted_partial_slot_b}
  \end{equation}
  Since for each of the matches $(k^\star,i^\star,t^\star)$, $(k^\star,j^\star,t^\star)$, $(k^\star+1,j^\star,t^\star)$ there exists a tournament containing exactly this match and in which team $t^\star$ never travels from venue $i^\star$ to venue $j^\star$, and since there exists a tournament satisfying condition~\ref{thm_travel_away_away_lifted_face_yes} of \cref{thm_travel_away_away_lifted_face}, we obtain
  \begin{equation*}
    \gamma = a_{k^\star,i^\star,t^\star} = a_{k^\star,j^\star,t^\star} = a_{k^\star+1,j^\star,t^\star} = \gamma = a_{k^\star,i^\star,t^\star} + a_{k^\star,j^\star,t^\star} - b_{t^\star,i^\star,j^\star}.
  \end{equation*}
  This shows that $\transpose{a}x + \transpose{b}y \leq \gamma$ is a positive multiple of inequality~\eqref{eq_travel_away_away_lifted_first}, which concludes the proof.
\end{proof}

Similar to \eqref{model_basic_travel_away_away}, inequalities~\eqref{model_basic_travel_home_away} are not facet-defining.
A lifted inequality reads
\begin{equation}
  x_{1,j,t} + x_{k,j,t} + \sum_{i \in V \setminus \{t\}} x_{k,t,i} + x_{k+1,j,t} - 1 \leq y_{t,t,j} \quad \forall k \in S  \setminus \{2n-2\}, ~\forall (t,j) \in A
  \label{eq_travel_home_away_lifted}
\end{equation}
Indeed, in order to obtain~\eqref{model_basic_travel_home_away} one only needs to combine~\eqref{eq_travel_home_away_lifted} with nonnegativity constraints for $x$.
The lifted inequalities turn out to be facet-defining.

\begin{theorem}
  \label{thm_travel_home_away}
  Inequalities~\eqref{eq_travel_home_away_lifted} are facet-defining for $\Ptt(n)$ for all $k \in S \setminus \{2n-2\}$ and $(t,j) \in A$.
\end{theorem}

\begin{proof}
  We assume $n \geq 6$ since we verified the statement for $n = 4$ computationally~\cite{IPO}.
  Consider the inequality $x_{1,j^\star,t^\star} + x_{k^\star,j^\star,t^\star} + \sum_{i \in V \setminus \{t^\star\}} x_{k^\star,t^\star,i} + x_{k^\star+1,j^\star,t^\star} - y_{t^\star,t^\star,j^\star} \leq 1$ for some slot $k^\star \in S \setminus \{2n-2\}$ and distinct teams $t^\star, j^\star \in V$.
  By \cref{thm_symmetry}, we can assume $j^\star = 3$ and $t^\star = 4$.
  The inequality is valid for $\Ptt(n)$ since the only possibilities in which $x_{1,j^\star,t^\star} + x_{k^\star,j^\star,t^\star} + \sum_{i \in V \setminus \{t^\star\}} x_{k^\star,t^\star,i} + x_{k^\star+1,j^\star,t^\star}$ exceeds $1$ are for $k^\star = 1$ (since then $(1,j^\star,t^\star)$ and $(k^\star,j^\star,t^\star)$ are identical) or if team $t^\star$ plays at home in slot $k^\star$ and away against team $j^\star$ in slot $1$ or $k^\star+1$.
  In either case, team $t^\star$ travels from its home venue to $j^\star$, forcing $y_{t^\star,t^\star,j^\star} = 1$.

  The following claim is used several times throughout the proof.

  \begin{restatable}{claim}{thmTravelHomeAwayLiftedFace}
    \label{thm_travel_home_away_lifted_face}
    Let $T$ be a tournament with 
    \begin{enumerate}[label={(\alph*)}]
    \item
      \label{thm_travel_home_away_lifted_face_init_first}
      $(1,j^\star,t^\star) \in T$ and $k^\star = 1$ holds, or
    \item
      \label{thm_travel_home_away_lifted_face_init_home}
      $(1,j^\star,t^\star) \in T$ and team $t^\star$ plays at home in slot $k^\star$, or
    \item
      \label{thm_travel_home_away_lifted_face_home_second}
      $(k^\star+1,j^\star,t^\star) \in T$ and team $t^\star$ plays at home in slot $k^\star$, or
    \item
      \label{thm_travel_home_away_lifted_face_away_second}
      $(k^\star+1,j^\star,t^\star) \in T$ and team $t^\star$ plays away in slot $k^\star$, or
    \item
      \label{thm_travel_home_away_lifted_face_away_first}
      $(k^\star,j^\star,t^\star) \in T$, $k^\star \geq 2$ and team $t^\star$ plays away in slot $k^\star-1$, or
    \item
      \label{thm_travel_home_away_lifted_face_home_only}
      team $t^\star$ plays at home in slot $k^\star$ and never travels from its home venue to venue $j^\star$.
    \end{enumerate}
    Then $(\chi(T),\psi(T))$ satisfies~\eqref{eq_travel_home_away_lifted} with equality.
    Moreover, team $t^\star$ travels from its home venue to venue $j^\star$ if and only if one of conditions~\ref{thm_travel_home_away_lifted_face_init_first}--\ref{thm_travel_home_away_lifted_face_home_second} is satisfied.
  \end{restatable}
  
  In order to prove that the inequality is facet-defining, let $\transpose{a}x + \transpose{b}y \leq \gamma$ define any facet $F$ that contains the face induced by this inequality.
  We will prove that it is a multiple of inequality~\eqref{eq_travel_home_away_lifted}.
  Let $\bar{k} \in S \setminus \{ 1, k^\star, k^\star+1 \}$.
  By \cref{thm_dim_basis} we can assume that the equation is normalized with respect to slot $\bar{k}$, i.e., it satisfies
  \begin{equation}
    a_m = 0 \text{ for each } m \in \columnBasis_{\bar{k}}.
    \tag{\S\ref*{thm_travel_home_away_lifted_face}}
    \label{proof_travel_home_away_lifted_basis}
  \end{equation}
  Note that, in contrast to previous proofs, we do not normalize with respect to slot $1$.

  \begin{restatable}{claim}{thmTravelHomeAwayLiftedNoTravel}
    \label{thm_travel_home_away_lifted_no_travel}
    For all $(t,i,j) \in V \times A$ with $(t,i,j) \neq (t^\star,t^\star,j^\star)$ there exists a tournament $T$ satisfying a condition from \cref{thm_travel_home_away_lifted_face}.
  \end{restatable}

  A tournament $T$ from \cref{thm_travel_home_away_lifted_no_travel} satisfies $\psi(T)_{t,i,j} = 0$.
  Let $y \coloneqq \psi(T)$ and let $y'$ be equal to $y$ except for $y'_{t,i,j} = 1$.
  By \cref{thm_travel_home_away_lifted_face} we have $(\chi(T),y), (\chi(T),y') \in F$.
  In this case, $\transpose{a}\chi(T) + \transpose{b}y = \gamma = \transpose{a}\chi(T) + \transpose{b}y'$ simplifies to
  \begin{equation}
    b_{t,i,j} = 0 \text{ for all } (t,i,j) \in V \times A \text{ with } (t,i,j) \neq (t^\star,t^\star,j^\star).
    \tag{\S\ref*{thm_travel_home_away_lifted_no_travel}}
    \label{proof_travel_home_away_lifted_no_travel}
  \end{equation}

  \begin{restatable}{claim}{thmTravelHomeAwayLiftedHomeAway}
    \label{thm_travel_home_away_lifted_home_away}
    For each $(k,i,j) \in \allMatches$ with $k \neq \bar{k}$, $\{i,j\} \neq \{j^\star,t^\star\}$ and for which $k = k^\star$ implies $t^\star \notin \{i,j\}$ there exist tournaments~$T$ and~$T'$ satisfying~\eqrefHomeAwaySwap{$\bar{k}$}{$k$}{$i$}{$j$} such that~$T$ and~$T'$ satisfy the same condition from \cref{thm_travel_home_away_lifted_face}.
  \end{restatable}

  The tournaments $T$ and $T'$ from \cref{thm_travel_home_away_lifted_home_away} satisfy $(\chi(T),\psi(T)), (\chi(T'),\psi(T')) \in F$ due to \cref{thm_travel_home_away_lifted_face}.
  Using \eqref{proof_travel_home_away_lifted_basis} and~\eqref{thm_travel_home_away_lifted_no_travel}, $\transpose{a}\chi(T) + \transpose{b}\psi(T) = \gamma = \transpose{a}\chi(T') + \transpose{b}\psi(T')$ simplifies to
  \begin{equation}
    a_{k,i,j} = a_{k,j,i} \text{ for each } (k,i,j) \in \allMatches \text{ with } \{i,j\} \neq \{j^\star,t^\star\} \text{ for which } k = k^\star \text{ implies } t^\star \notin \{i,j\}.
    \tag{\S\ref*{thm_travel_home_away_lifted_home_away}}
    \label{proof_travel_home_away_lifted_home_away}
  \end{equation}
  
  \begin{restatable}{claim}{thmTravelHomeAwayLiftedPartialSlot}
    \label{thm_travel_home_away_lifted_partial_slot}
    Let $k \in S \setminus \{\bar{k}\}$, let $i,j,i',j' \in V$ be distinct with $(j^\star,t^\star) \notin \{ (i,j), (i',j'), (i,j'), (i',j) \}$ or with $k \notin \{1,k^\star,k^\star+1\}$.
    Then there exist tournaments $T$ and $T'$ satisfying~\eqrefPartialSlotSwap{$\bar{k}$}{$k$}{$i$}{$j$}{$i'$}{$j'$} such that~$T$ and~$T'$ satisfy the same condition from \cref{thm_travel_home_away_lifted_face}.
  \end{restatable}
  
  The tournaments $T$ and $T'$ from \cref{thm_travel_home_away_lifted_partial_slot} satisfy $(\chi(T),\psi(T)), (\chi(T'),\psi(T')) \in F$ due to \cref{thm_travel_home_away_lifted_face}.
  Using \eqref{proof_travel_home_away_lifted_basis} and~\eqref{thm_travel_home_away_lifted_no_travel}, $\transpose{a}\chi(T) + \transpose{b}\psi(T) = \gamma = \transpose{a}\chi(T') + \transpose{b}\psi(T')$ simplifies to
  \begin{multline}
    a_{k,i,j} + a_{k,i',j'} = a_{k,i,j'} + a_{k,i',j} \text{ for all distinct } i,j,i',j' \in V \text{ with } \\ (j^\star,t^\star) \notin \{ (i,j), (i',j'), (i,j'), (i',j) \} \text{ or } k \notin \{ 1, k^\star, k^\star+1 \} .
    \tag{\S\ref*{thm_travel_home_away_lifted_partial_slot}a}
    \label{proof_travel_home_away_lifted_partial_slot_a}
  \end{multline}
  Consider a slot $k \in S \setminus \{\bar{k}\}$.
  For each $\ell \in \{4,5,\dotsc,n\}$, \eqref{proof_travel_home_away_lifted_partial_slot_a} for $(i,j,i',j') = (1,\ell,2,3)$ is applicable since $(j^\star,t^\star) = (3,4)$ is not among the matches $(i,j)$, $(i',j')$, $(i,j')$, $(i',j)$.
  This implies $a_{k,1,\ell} + a_{k,2,3} = a_{k,1,3} + a_{k,2,\ell}$ which together with \eqref{proof_travel_home_away_lifted_basis} yields $a_{k,2,\ell} = 0$.
  Moreover, for each $\ell \in \{3,4,\dotsc,n\}$ with $(k,\ell) \neq (k^\star,t^\star)$, \eqref{proof_travel_home_away_lifted_home_away} for $(i,j) = (\ell,2)$ implies $a_{k,\ell,2} = a_{k,2,\ell} = 0$.

  For distinct $\ell, \ell' \in \{3,4,\dotsc,n\}$ with $(\ell,\ell') \neq (3,4)$ or $k \notin \{1,k^\star.k^\star+1\}$, \eqref{proof_travel_home_away_lifted_partial_slot_a} for $(i,j,i',j') = (1,\ell',\ell,2)$ is applicable, which implies $a_{k,1,\ell'} + a_{k,\ell,2} = a_{k,1,2} + a_{k,\ell,\ell'}$.
  Together with \eqref{proof_travel_home_away_lifted_basis} this shows
  \begin{multline}
    a_{k,i,j} = 0 \text{ for all } (k,i,j) \in \allMatches \text{ with } \\ (k,i) \neq (k^\star,t^\star) \text{ and for which } (i,j) = (j^\star,t^\star) \text{ implies } k \notin \{ 1, k^\star, k^\star+1 \}.
    \tag{\S\ref*{thm_travel_home_away_lifted_partial_slot}b}
    \label{proof_travel_home_away_lifted_partial_slot_b}
  \end{multline}
  Together with~\eqref{proof_travel_home_away_lifted_no_travel}, we obtain that the support of inequality $\transpose{a}x + \transpose{b}y \leq \gamma$ is a subset of the support of inequality~\eqref{eq_travel_home_away_lifted}.
  It remains to prove that the coefficients agree (up to a positive multiple).

  It is easy to see that for each condition of \cref{thm_travel_home_away_lifted_face} there exists a tournament $T$ satisfying it.
  From~\eqref{proof_travel_home_away_lifted_no_travel} and~\eqref{proof_travel_home_away_lifted_partial_slot_b} we obtain the following equations:
  If $k^\star = 1$, then
  \begin{equation*}
    \gamma \overset{\ref{thm_travel_home_away_lifted_face_init_first}}{=} a_{1,j^\star,t^\star} - y_{t^\star,t^\star,j^\star} \overset{\ref{thm_travel_home_away_lifted_face_home_second}}{=} a_{k^\star+1,j^\star,t^\star} + a_{k^\star,t^\star,j} - y_{t^\star,t^\star,j^\star} \overset{\ref{thm_travel_home_away_lifted_face_away_second}}{=} a_{k^\star+1,j^\star,t^\star} \overset{\ref{thm_travel_home_away_lifted_face_home_only}}{=} a_{k^\star,t^\star,j}
  \end{equation*}
  holds, which implies $a_{1,j^\star,t^\star} = 2$ and $a_{1,t^\star,j} = a_{2,j^\star,t^\star} = b_{t^\star,t^\star,j^\star} = \gamma = 1$ for each $j \in V \setminus \{t^\star\}$.
  Otherwise, i.e., if $k^\star \geq 2$, then 
  \begin{equation*}
    \gamma \overset{\ref{thm_travel_home_away_lifted_face_init_home}}{=} a_{1,j^\star,t^\star} + a_{k^\star,t^\star,j} - y_{t^\star,t^\star,j^\star} \overset{\ref{thm_travel_home_away_lifted_face_home_second}}{=} a_{k^\star+1,j^\star,t^\star} + a_{k^\star,t^\star,j} - y_{t^\star,t^\star,j^\star} \overset{\ref{thm_travel_home_away_lifted_face_away_second}}{=} a_{k^\star+1,j^\star,t^\star} \overset{\ref{thm_travel_home_away_lifted_face_home_second}}{=} a_{k^\star,j^\star,t^\star} \overset{\ref{thm_travel_home_away_lifted_face_home_only}}{=} a_{k^\star,t^\star,j}       
  \end{equation*}
  holds, which implies $a_{1,j^\star,t^\star} = a_{k^\star,j^\star,t^\star} = a_{k^\star,t^\star,j} = a_{k^\star+1,j^\star,t^\star} = b_{t^\star,t^\star,j^\star} = \gamma = 1$ for each $j \in V \setminus \{t^\star\}$.
  This shows that $\transpose{a}x + \transpose{b}y \leq \gamma$ is a positive multiple of inequality~\eqref{eq_travel_home_away_lifted}, which concludes the proof.
\end{proof}

The symmetric lifted version of inequality~\eqref{model_basic_travel_away_home} reads
\begin{equation}
  x_{2n-2,i,t} + x_{k,i,t} + \sum_{j \in V \setminus \{t\}} x_{k,t,j} + x_{k-1,i,t} - 1 \leq y_{t,i,t} \quad \forall k \in S  \setminus \{1\}, ~\forall (i,t) \in A
  \label{eq_travel_away_home_lifted}
\end{equation}
Using \cref{thm_symmetry}, we obtain the following corollary of \cref{thm_travel_home_away}.

\begin{corollary}
  Inequalities~\eqref{eq_travel_away_home_lifted} are facet-defining for $\Ptt(n)$ for all $k \in S \setminus \{2n-2\}$ and $(i,t) \in A$.
\end{corollary}

\begin{theorem}
  \label{thm_travel_first}
  Inequalities~\eqref{model_basic_travel_first}, $x_{1,j,t} \leq y_{t,t,j}$, are facet-defining for $\Ptt(n)$ for all $(t,j) \in A$.
\end{theorem}

\begin{proof}
  We assume $n \geq 6$ since we verified the statement for $n = 4$ computationally~\cite{IPO}.
  Consider the inequality $x_{1,j^\star,t^\star} \leq y_{t^\star,t^\star,j^\star}$ for distinct teams $t^\star, j^\star \in V$.
  By \cref{thm_symmetry}, we can assume $j^\star = 3$ and $t^\star = 4$.
  The inequality is valid for $\Ptt(n)$ since the team $t^\star$ has to travel from its home venue to venue $j^\star$ if it plays there in slot $1$.

  The following claim is used several times throughout the proof.
  \begin{restatable}{claim}{thmTravelFirstFace}
    \label{thm_travel_first_face}
    Let $T$ be a tournament
    \begin{enumerate}[label={(\alph*)}]
    \item
      \label{thm_travel_first_face_no}
      in which team $t^\star$ never travels from its home venue to venue $j^\star$, or
    \item
      \label{thm_travel_first_face_yes}
      with $(1,j^\star,t^\star) \in T$.
    \end{enumerate}
    Then $(\chi(T),\psi(T))$ satisfies~\eqref{model_basic_travel_first} with equality.
    Moreover, team $t^\star$ travels from its home venue to venue $j^\star$ if and only if condition~\ref{thm_travel_first_face_yes} is satisfied.
  \end{restatable}
  
  In order to prove that the inequality is facet-defining, let $\transpose{a}x + \transpose{b}y \leq \gamma$ define any facet $F$ that contains the face induced by this inequality.
  We will prove that it is a multiple of inequality~\eqref{model_basic_travel_first}.
  
  By \cref{thm_dim_basis} we can assume that the equation is normalized with respect to slot $1$, i.e., it satisfies~\eqref{eq_basis_zero}.

  \begin{restatable}{claim}{thmTravelFirstNoTravel}
    \label{thm_travel_first_no_travel}
    For all $(t,i,j) \in V \times A$ with $(t,i,j) \neq (t^\star,t^\star,j^\star)$ there exists a tournament $T$ satisfying a condition of \cref{thm_travel_first_face}.
  \end{restatable}

  A tournament $T$ from \cref{thm_travel_first_no_travel} satisfies $\psi(T)_{t,i,j} = 0$.
  Let $y \coloneqq \psi(T)$ and let $y'$ be equal to $y$ except for $y'_{t,i,j} = 1$.
  By \cref{thm_travel_first_face} we have $(\chi(T),y), (\chi(T),y') \in F$.
  In this case, $\transpose{a}\chi(T) + \transpose{b}y = \gamma = \transpose{a}\chi(T) + \transpose{b}y'$ simplifies to
  \begin{equation}
    b_{t,i,j} = 0 \text{ for all } (t,i,j) \in V \times A \text{ with } (t,i,j) \neq (t^\star,t^\star,j^\star).
    \tag{\S\ref*{thm_travel_first_no_travel}}
    \label{proof_travel_first_no_travel}
  \end{equation}

  \begin{restatable}{claim}{thmTravelFirstHomeAway}
    \label{thm_travel_first_home_away}
    For each $(k,i,j) \in \allMatches$ with $k \neq n$ and $\{i,j\} \neq \{j^\star,t^\star\}$ there exist tournaments~$T$ and~$T'$ satisfying~\eqrefHomeAwaySwap{$n$}{$k$}{$i$}{$j$} such that~$T$ and~$T'$ satisfy the same condition from \cref{thm_travel_first_face}.
  \end{restatable}

  The tournaments $T$ and $T'$ from \cref{thm_travel_first} satisfy $(\chi(T),\psi(T)), (\chi(T'),\psi(T')) \in F$ by \cref{thm_travel_first_face}.
  Using \eqref{eq_basis_zero} and~\eqref{thm_travel_first_no_travel}, $\transpose{a}\chi(T) + \transpose{b}\psi(T) = \gamma = \transpose{a}\chi(T') + \transpose{b}\psi(T')$ simplifies to
  \begin{equation}
    a_{k,i,j} = a_{k,j,i} \text{ for each } (k,i,j) \in \allMatches \text{ with } \{i,j\} \neq \{j^\star,t^\star\}.
    \tag{\S\ref*{thm_travel_first_home_away}}
    \label{proof_travel_first_home_away}
  \end{equation}
  
  \begin{restatable}{claim}{thmTravelFirstPartialSlot}
    \label{thm_travel_first_partial_slot}
    Let $k \in S \setminus \{n\}$, let $i,j,i',j' \in V$ be distinct such that $(k,j^\star,t^\star) \notin \{ (1,i,j)$, $(1,i',j')$, $(1,i,j')$, $(1,i',j) \}$ holds.
    Then there exist tournaments $T$ and $T'$ satisfying~\eqrefPartialSlotSwap{$n$}{$k$}{$i$}{$j$}{$i'$}{$j'$} such that~$T$ and~$T'$ satisfy the same condition from \cref{thm_travel_first_face}.
  \end{restatable}
  
  The tournaments $T$ and $T'$ from \cref{thm_travel_first_partial_slot} satisfy $(\chi(T),\psi(T)), (\chi(T'),\psi(T')) \in F$ due to \cref{thm_travel_first_face}.
  Using \eqref{eq_basis_zero} and~\eqref{thm_travel_first_no_travel}, $\transpose{a}\chi(T) + \transpose{b}\psi(T) = \gamma = \transpose{a}\chi(T') + \transpose{b}\psi(T')$ simplifies to
  \begin{multline}
    a_{k,i,j} + a_{k,i',j'} = a_{k,i,j'} + a_{k,i',j} \text{ for all distinct } i,j,i',j' \in V \text{ with } \\ (j^\star,t^\star) \notin \{ (i,j), (i',j'), (i,j'), (i',j) \}.
    \tag{\S\ref*{thm_travel_first_partial_slot}a}
    \label{proof_travel_first_partial_slot_a}
  \end{multline}

  Consider a slot $k \in S \setminus \{n\}$.
  For each $\ell \in \{4,5,\dotsc,n\}$, \eqref{proof_travel_first_partial_slot_a} for $(i,j,i',j') = (1,\ell,2,3)$ is applicable since $(j^\star,t^\star) = (3,4)$ is not among the matches $(i,j)$, $(i',j')$, $(i,j')$, $(i',j)$.
  This implies $a_{k,1,\ell} + a_{k,2,3} = a_{k,1,3} + a_{k,2,\ell}$ which together with \eqref{eq_basis_zero} yields $a_{k,2,\ell} = 0$.
  Moreover, for each $\ell \in \{3,4,\dotsc,n\}$, \eqref{proof_travel_first_home_away} for $(i,j) = (\ell,2)$ implies $a_{k,\ell,2} = a_{k,2,\ell} = 0$.

  For distinct $\ell, \ell' \in \{3,4,\dotsc,n\}$ with $(k,\ell,\ell') \neq (1,3,4)$, \eqref{proof_travel_first_partial_slot_a} for $(i,j,i',j') = (1,\ell',\ell,2)$ is applicable, which implies $a_{k,1,\ell'} + a_{k,\ell,2} = a_{k,1,2} + a_{k,\ell,\ell'}$.
  Together with \eqref{eq_basis_zero} this shows
  \begin{equation}
    a_{k,i,j} = 0 \text{ for all } (k,i,j) \in \allMatches \setminus \{ (1,j^\star,t^\star) \}.
    \tag{\S\ref*{thm_travel_first_partial_slot}b}
    \label{proof_travel_first_partial_slot_b}
  \end{equation}
  Together with~\eqref{proof_travel_first_no_travel}, we obtain that the support of inequality $\transpose{a}x + \transpose{b}y \leq \gamma$ is a subset of the support of inequality~\eqref{model_basic_travel_first}.

  It remains to prove that the coefficients agree (up to a positive multiple).
  From \cref{thm_travel_first_face} it is clear that $a_{1,j^\star,t^\star} = - b_{t^\star,t^\star,j^\star}$ and that the right-hand side $\gamma$ must be equal to $0$.
  This concludes the proof.
\end{proof}

Again, we obtain the following corollary by applying \cref{thm_symmetry}.

\begin{corollary}
  Inequalities~\eqref{model_basic_travel_last}, $x_{2n-2,i,t} \leq y_{t,i,t}$, are facet-defining for $\Ptt(n)$ for all $(i,t) \in A$.
\end{corollary}

\section{New inequality classes}
\label{sec_strengthening_inequalities}

\paragraph{Flow inequalities.}

Formulation~\eqref{model_basic} can be strengthened by the following \emph{flow inequalities}.
\begin{subequations}
  \label{eq_flow_distinct}
  \begin{align}
    \sum_{j \in V \setminus \{i\}} y_{t,i,j} &\geq 1 \quad \forall i,t \in V : i \neq t \label{eq_flow_out_distinct} \\
    \sum_{j \in V \setminus \{i\}} y_{t,j,i} &\geq 1 \quad \forall i,t \in V : i \neq t \label{eq_flow_in_distinct}
  \end{align}
\end{subequations}
They state that each team $t$ has to leave (resp.\ enter) each other team's venue at least once.
We now prove that all these inequalities define facets of $\Ptt(n)$.

\begin{theorem}
  \label{thm_flow_distinct}
  Inequalities~\eqref{eq_flow_distinct} are facet-defining for $\Ptt(n)$ for all $i,t \in V$ with $i \neq t$.
\end{theorem}

\begin{proof}
  We only prove the statement for inequalities~\eqref{eq_flow_out_distinct}.
  For~\eqref{eq_flow_in_distinct}, it then follows from \cref{thm_symmetry}.
  In addition, we assume $n \geq 8$ since we verified the statement for $n \in \{4,6\}$ computationally~\cite{IPO}.

  Let $i^\star, t^\star \in V$ with $i^\star \neq t^\star$.
  The inequality for $i \coloneqq i^\star$ and $t \coloneqq t^\star$ is valid since team $t^\star$ has to play an away match against team $i^\star$ after which it leaves to some other venue.

  To establish that the inequality is facet-defining, let $\transpose{a} x + \transpose{b}y \geq \gamma$ define any facet $F$ that contains the face induced by $\sum_{j \in V \setminus \{i^\star\}} y_{t^\star,i^\star,j} \geq 1$.
  Without loss of generality, the equation is normalized with respect to slot $1$, i.e., it satisfies~\eqref{eq_basis_zero}.

  \begin{restatable}{claim}{thmFlowOutDistinctNoTravel}
    \label{thm_flow_out_distinct_no_travel}
    For all $(t,i,j) \in V \times A$ with $(t,i) \neq (t^\star,i^\star)$ there exists a tournament in which team $t$ never travels from venue $i$ to venue $j$ and in which team $t^\star$ leaves venue $i^\star$ exactly once.
  \end{restatable}

  A tournament $T$ from \cref{thm_flow_out_distinct_no_travel} satisfies $\psi(T)_{t,i,j} = 0$.
  Let $y \coloneqq \psi(T)$ and let $y'$ be equal to $y$ except for $y'_{t,i,j} = 1$.
  We have $(\chi(T),y) \in F$ and if $(t,i) \neq (t^\star,i^\star)$ holds, also $(\chi(T),y') \in F$.
  In this case, $\transpose{a}\chi(T) + \transpose{b}y = \gamma = \transpose{a}\chi(T) + \transpose{b}y'$ simplifies to $b_{t,i,j} = 0$.
  We obtain
  \begin{equation}
    b_{t,i,j} = 0 \text{ for all } (t,i,j) \in V \times A \text{ with } (t,i) \neq (t^\star,i^\star).
    \tag{\S\ref*{thm_flow_out_distinct_no_travel}}
    \label{proof_flow_out_distinct_no_travel}
  \end{equation}
  
  \begin{restatable}{claim}{thmFlowOutDistinctHomeAway}
    \label{thm_flow_out_distinct_home_away}
    For all distinct $i,j \in V$ and for each $k \in S \setminus \{1\}$ there exist tournaments $T$ and $T'$ satisfying~\eqrefHomeAwaySwapDefault{} such that in both tournaments team $t^\star$ leaves venue $i^\star$ exactly once and to the same venue.
  \end{restatable}

  In the tournaments $T$ and $T'$ from \cref{thm_flow_out_distinct_home_away}
  team $t^\star$ leaves venue $i^\star$ exactly once and to the same venue.
  Hence, we have $(\chi(T),\psi(T)), (\chi(T'),\psi(T')) \in F$.
  Moreover, together with \eqref{proof_flow_out_distinct_no_travel} it implies $\transpose{b}\psi(T) = \transpose{b}\psi(T')$.
  Combining this with the fact that the equation is normalized with respect to slot $1$ (i.e., it satisfies~\eqref{eq_basis_zero}), $\transpose{a}\chi(T) + \transpose{b}\psi(T) = \gamma = \transpose{a}\chi(T') + \transpose{b}\psi(T')$ simplifies to $a_{k,i,j} = a_{k,j,i}$.
  Thus, we have
  \begin{equation}
    a_{k,i,j} = a_{k,j,i} \text{ for each } (k,i,j) \in \allMatches.
    \tag{\S\ref*{thm_flow_out_distinct_home_away}}
    \label{proof_flow_out_distinct_home_away}
  \end{equation}

  \begin{restatable}{claim}{thmFlowOutDistinctPartialSlot}
    \label{thm_flow_out_distinct_partial_slot}
    For each slot $k \in S \setminus \{1\}$ and for distinct teams $i,j,i',j' \in V$ with $(i^\star,t^\star) \notin \{(i,j)$,$(i',j')$, $(i,j')$, $(i',j)\}$ there exist tournaments $T$ and $T'$ satisfying~\eqrefPartialSlotSwapDefault{} such that in both tournaments team $t^\star$ leaves venue $i^\star$ exactly once and to the same venue.
  \end{restatable}

  In the tournaments $T$ and $T'$ from \cref{thm_flow_out_distinct_partial_slot}
  team $t^\star$ leaves venue $i^\star$ exactly once and to the same venue.
  Hence, we have $(\chi(T),\psi(T)), (\chi(T'),\psi(T')) \in F$.
  Moreover, together with \eqref{proof_flow_out_distinct_no_travel} it implies $\transpose{b}\psi(T) = \transpose{b}\psi(T')$.
  Combining this with the fact that the equation is normalized with respect to slot $1$ (i.e., it satisfies~\eqref{eq_basis_zero}), $\transpose{a}\chi(T) + \transpose{b}\psi(T) = \gamma = \transpose{a}\chi(T') + \transpose{b}\psi(T')$ simplifies to
  \begin{multline}
    a_{k,i,j} + a_{k,i',j'} = a_{k,i,j'} + a_{k,i',j} \text{ for each } k \in S \setminus \{1\} \text{ and for all distinct } i,j,i',j' \in V \\ \text{ with } (i^\star,t^\star) \notin \{ (i,j),(i',j'),(i,j'),(i',j) \}
    \tag{\S\ref*{thm_flow_out_distinct_partial_slot}}
    \label{proof_flow_out_distinct_partial_slot}
  \end{multline}

  Since the formulation is symmetric with respect to teams, we can now, by permuting teams, assume $(i^\star,t^\star) = (4,3)$.
  Consider a slot $k \in S \setminus \{1\}$.
  For each $\ell \in \{4,5,\dotsc,n\}$, \eqref{proof_flow_out_distinct_partial_slot} implies $a_{k,1,\ell} + a_{k,2,3} = a_{k,1,3} + a_{k,2,\ell}$ which together with the fact that the equation is normalized with respect to slot $1$ (i.e., it satisfies~\eqref{eq_basis_zero}) yields $a_{k,2,\ell} = 0$.
  Combined with~\eqref{proof_flow_out_distinct_home_away} we also obtain $a_{k,\ell,2} = 0$.
  For all $\ell,\ell' \in \{3,4,\dotsc,n\}$ except for $(\ell,\ell') = (4,3)$, \eqref{proof_flow_out_distinct_partial_slot} implies $a_{k,1,\ell'} + a_{k,\ell,2} = a_{k,1,2} + a_{k,\ell,\ell'}$.
  Together with the fact that the equation is normalized with respect to slot $1$ (i.e., it satisfies~\eqref{eq_basis_zero}), this shows $a_{k,\ell,\ell'} = 0$ for all $(\ell,\ell') \neq (4,3)$.
  From \eqref{proof_flow_out_distinct_home_away} we also have $a_{k,4,3} = a_{k,3,4} = 0$ and obtain $a = \zerovec$.

  \begin{restatable}{claim}{thmFlowOutDistinctCoefficients}
    \label{thm_flow_out_distinct_coefficients}
    For distinct $j,j' \in V \setminus \{i^\star\}$ there exist tournaments $T$ and $T'$ such that in both tournaments team $t^\star$ leaves venue $i^\star$ exactly once, namely to venue $j$ in $T$ and to venue $j'$ in $T'$.
  \end{restatable}

  In the tournaments $T$ and $T'$ from \cref{thm_flow_out_distinct_coefficients}
  team $t^\star$ leaves venue $i^\star$ exactly once.
  Hence, we have $(\chi(T),\psi(T)), (\chi(T'),\psi(T')) \in F$.
  From $a = \zerovec$ and~\eqref{proof_flow_out_distinct_no_travel} we have that $b_{t^\star,i^\star,j} = \transpose{a}\chi(T) + \transpose{b}\psi(T) = \gamma = \transpose{a}\chi(T') + \transpose{b}\psi(T') = b_{t^\star,i^\star,j'}$.
  This shows that $(\transpose{a}, \transpose{b})$ is a multiple of the coefficient vector of~\eqref{eq_flow_out_distinct}.
  The fact that it is a positive multiple follows from the observation that we can take any feasible solution and setting all entries of $y$ to $1$ yields another feasible solution (which is not in the face anymore).
\end{proof}

\paragraph{Home-flow inequalities.}
Inequalities~\eqref{eq_flow_distinct} also hold for $t$'s home venue, i.e., $i = t$, but in this case they are dominated by the following \emph{home-flow inequalities}.
\begin{subequations}
  \label{eq_flow_home}
  \begin{align}
    \sum_{j \in V \setminus \{t\}} y_{t,t,j} + \sum_{j \in V \setminus \{t\}} (x_{k,t,j} + x_{k+n-1,t,j}) &\geq 2 \quad \forall k \in \{1,2,\dotsc,n-1\},~ \forall t \in V
    \label{eq_flow_out_home} \\
    \sum_{j \in V \setminus \{t\}} y_{t,t,j} + \sum_{j \in V \setminus \{t\}} (x_{k,j,t} + x_{k+n-1,j,t}) &\geq 2 \quad \forall k \in \{1,2,\dotsc,n-1\},~ \forall t \in V
    \label{eq_flow_out_away} \\
    \sum_{i \in V \setminus \{t\}} y_{t,i,t} + \sum_{i \in V \setminus \{t\}} (x_{k,t,i} + x_{k+n-1,t,i}) &\geq 2 \quad \forall k \in \{1,2,\dotsc,n-1\},~ \forall t \in V
    \label{eq_flow_in_home} \\
    \sum_{i \in V \setminus \{t\}} y_{t,i,t} + \sum_{i \in V \setminus \{t\}} (x_{k,i,t} + x_{k+n-1,i,t}) &\geq 2 \quad \forall k \in \{1,2,\dotsc,n-1\},~ \forall t \in V
    \label{eq_flow_in_away}
  \end{align}
\end{subequations}
They are valid for $\Ptt(n)$ since team $t$ leaves (resp.\ enters) its home venue either at least twice or it leaves (resp.\ enters) it only once in which case it cannot play at home (resp.\ away) in slots $k$ and $k+n-1$.
The sum of the first two reads
\begin{equation*}
  \sum_{j \in V \setminus \{t\}} 2y_{t,t,j} + \sum_{j \in V \setminus \{t\}} (x_{k,t,j} + x_{k,j,t} + x_{k+n-1,t,j} + x_{k+n-1,j,t}) \geq 4,
\end{equation*}
for which the subtraction of equation~\eqref{model_basic_team_plays} for team $t$ and slots $k$ and $k+n-1$ yields
\begin{equation*}
  \sum_{j \in V \setminus \{t\}} 2y_{t,t,j} \geq 4 - 1 - 1,
\end{equation*}
which in turn equals~\eqref{eq_flow_out_distinct} for $i = t$.
The corresponding result is as follows.

\begin{theorem}
  \label{thm_flow_home_away}
  Inequalities~\eqref{eq_flow_home} are facet-defining for $\Ptt(n)$ for each team $t \in V$ and each slot $k \in \{1,2,\dotsc,n-1\}$.
\end{theorem}

\begin{proof}
  We only prove the statement for inequalities~\eqref{eq_flow_out_home}.
  The proof for inequalities~\eqref{eq_flow_out_away} is very similar.
  Moreover, the result for inequalities~\eqref{eq_flow_in_home} and~\eqref{eq_flow_in_away} follows from \cref{thm_symmetry}.
  In addition, we assume $n \geq 6$ since we verified the statement for $n = 4$ computationally~\cite{IPO}.

  Let $t^\star \in V$ and $k^\star \in \{1,2,\dotsc,n-1\}$.
  To see that the inequalities are valid, first observe that team $t^\star$ has to leave its own venue at least once.
  If it does so at least twice, the inequality is certainly satisfied.
  The remaining case is settled by the following observation which we will use several times throughout the proof.

  \begin{claim}
    \label{thm_flow_out_home_away_leave_once}
    Let $T$ be a tournament in which $t^\star$ leaves its home venue exactly once.
    Then all away matches of $t^\star$ take place in consecutive slots, and hence $t^\star$ plays at home in exactly one of the two slots $k^\star$ and $k^\star+n-1$.
    In particular, $(\chi(T),\psi(T))$ satisfies~\eqref{eq_flow_out_home} and~\eqref{eq_flow_out_away} with equality.
  \end{claim}

  To prove that inequality~\eqref{eq_flow_out_home} is facet-defining, let $\transpose{a} x + \transpose{b}y \geq \gamma$ define any facet $F$ that contains the face induced by $\sum_{j \in V \setminus \{t^\star\}} y_{t^\star,t^\star,j} + \sum_{j \in V \setminus \{t^\star\}} (x_{k^\star,t^\star,j} + x_{k^\star+n-1,t^\star,j}) \geq 2$.

  Since the formulation is symmetric with respect to teams we can, by permuting teams, assume $t^\star = 4$ for the remainder of the proof.
  Let $\bar{k} \in S$ with $k^\star < \bar{k} < k^\star + n-1$.
  By \cref{thm_dim_basis} we can assume that the equation is normalized with respect to slot $\bar{k}$, i.e., it satisfies
  \begin{equation}
    a_m = 0 \text{ for each } m \in \columnBasis_{\bar{k}}.
    \tag{\S\ref*{thm_flow_out_home_away_leave_once}}
    \label{proof_flow_out_home_basis}
  \end{equation}

  \begin{restatable}{claim}{thmFlowOutHomeNoTravel}
    \label{thm_flow_out_home_no_travel}
    For all $(t,i,j) \in V \times A$ with $(t,i) \neq (t^\star,t^\star)$ there exists a tournament in which team $t$ never travels from venue $i$ to venue $j$ and in which team $t^\star$ leaves its home venue exactly once.
  \end{restatable}

  A tournament $T$ from \cref{thm_flow_out_home_no_travel} satisfies $\psi(T)_{t,i,j} = 0$.
  Let $y \coloneqq \psi(T)$ and let $y'$ be equal to $y$ except for $y'_{t,i,j} = 1$.
  By \cref{thm_flow_out_home_away_leave_once} we have $(\chi(T),y) \in F$ and if $(t,i) \neq (t^\star,t^\star)$ holds, also $(\chi(T),y') \in F$.
  In this case, $\transpose{a}\chi(T) + \transpose{b}y = \gamma = \transpose{a}\chi(T) + \transpose{b}y'$ simplifies to $b_{t,i,j} = 0$.
  We obtain
  \begin{equation}
    b_{t,i,j} = 0 \text{ for all } (t,i,j) \in V \times A \text{ with } (t,i) \neq (t^\star,t^\star).
    \tag{\S\ref*{thm_flow_out_home_no_travel}}
    \label{proof_flow_out_home_no_travel}
  \end{equation}

  \begin{restatable}{claim}{thmFlowOutHomeCofficients}
    \label{thm_flow_out_home_coefficients}
    For any slot $k \in \{1,2,\dotsc,n-1\}$ and distinct $j,j' \in V \setminus \{t^\star\}$ there exist tournaments $T$ and $T'$ satisfying~\eqrefHomeAwaySwap{$k$}{$k+n-1$}{$t^\star$}{$j$} and such that team $t^\star$ leaves its home venue exactly once and to the venues $j$ in $T$ and to $j'$ in $T'$.
  \end{restatable}

  In the tournaments $T$ and $T'$ from \cref{thm_flow_out_home_coefficients}
  team $t^\star$ leaves its home venue exactly once.
  Hence, by \cref{thm_flow_out_home_away_leave_once} we have $(\chi(T),\psi(T)), (\chi(T'),\psi(T')) \in F$.
  Due to~\eqref{proof_flow_out_home_basis} and \eqref{proof_flow_out_home_no_travel} the equation $\transpose{a}\chi(T) + \transpose{b}\psi(T) = \gamma = \transpose{a}\chi(T') + \transpose{b}\psi(T')$ simplifies to $a_{k,t^\star,j} + a_{k+n-1,j,t^\star} + b_{t^\star,t^\star,j} = a_{k,j,t^\star} + a_{k+n-1,t^\star,j} + b_{t^\star,t^\star,j'}$.
  Since $j'$ only appears in the last term, varying $j'$ yields
  $b_{t^\star,t^\star,j_1} = b_{t^\star,t^\star,j_2}$ for all $j_1,j_2 \in V \setminus \{t^\star\}$.
  Together with~\eqref{proof_flow_out_home_no_travel}, this shows
  \begin{multline}
    \transpose{b}\psi(T) = \transpose{b}\psi(T') \text{ for all tournaments } T,T' \text{ with } (\chi(T),\psi(T)),(\chi(T'),\psi(T')) \in F \\
    \text{in which } t^\star \text{ leaves its home venue as often in $T$ as in $T'$}.
    \tag{\S\ref*{thm_flow_out_home_coefficients}a}
    \label{proof_flow_out_home_coefficients}
  \end{multline}
  This further simplifies the equation to
  \begin{equation}
    a_{k,t^\star,j} + a_{k+n-1,j,t^\star} = a_{k,j,t^\star} + a_{k+n-1,t^\star,j} \text{ for all } k \in \{1,2,\dotsc,n-1\} \text{ and all } j \in V \setminus \{t^\star\}.
    \tag{\S\ref*{thm_flow_out_home_coefficients}b}
    \label{proof_flow_out_home_home_away_special}
  \end{equation}
  
  \begin{restatable}{claim}{thmFlowOutHomeHomeAway}
    \label{thm_flow_out_home_home_away}
    For each slot $k \in S \setminus \{\bar{k}\}$ and for all distinct $i,j \in V \setminus \{t^\star\}$ there exist tournaments $T$ and $T'$ satisfying~\eqrefHomeAwaySwap{$\bar{k}$}{$k$}{$i$}{$j$} and such that in both tournaments team $t^\star$ leaves its home venue exactly once.
  \end{restatable}

  In the tournaments $T$ and $T'$ from \cref{thm_flow_out_home_home_away} team $t^\star$ leaves its home venue exactly once.
  Hence, by \cref{thm_flow_out_home_away_leave_once} we have $(\chi(T),\psi(T)), (\chi(T'),\psi(T')) \in F$ and by~\eqref{proof_flow_out_home_coefficients} also $\transpose{b}\psi(T) = \transpose{b}\psi(T')$.
  Combining this with~\eqref{proof_flow_out_home_basis}, $\transpose{a}\chi(T) + \transpose{b}\psi(T) = \gamma = \transpose{a}\chi(T') + \transpose{b}\psi(T')$ simplifies to $a_{k,i,j} = a_{k,j,i}$.
  Thus, we have
  \begin{equation}
    a_{k,i,j} = a_{k,j,i} \text{ for each } (k,i,j) \in \allMatches \text{ with } t^\star \notin \{i,j\}.
    \tag{\S\ref*{thm_flow_out_home_home_away}}
    \label{proof_flow_out_home_home_away}
  \end{equation}

  \begin{restatable}{claim}{thmFlowOutHomePartialSlot}
    \label{thm_flow_out_home_partial_slot}
    For distinct slots $k_1,k_2 \in S$ and distinct teams $i,j,i',j' \in V$ with $t^\star \notin \{i,i'\}$ and with $k_2 = k_1 + 1$ if $t^\star \in \{j,j'\}$ there exist tournaments $T$ and $T'$ satisfying~\eqrefPartialSlotSwap{$k_1$}{$k_2$}{$i$}{$j$}{$i'$}{$j'$} such that in both tournaments team $t^\star$ leaves its home venue exactly once.
  \end{restatable}

  In the tournaments $T$ and $T'$ from \cref{thm_flow_out_home_partial_slot} team $t^\star$ leaves its home venue exactly once.
  Hence, by \cref{thm_flow_out_home_away_leave_once}  we have $(\chi(T),\psi(T)), (\chi(T'),\psi(T')) \in F$ and by~\eqref{proof_flow_out_home_coefficients} also $\transpose{b}\psi(T) = \transpose{b}\psi(T')$.
  Combining this with~\eqref{proof_flow_out_home_basis}, equation $\transpose{a}\chi(T) + \transpose{b}\psi(T) = \gamma = \transpose{a}\chi(T') + \transpose{b}\psi(T')$ yields
  \begin{multline}
    a_{k_1,i,j} + a_{k_1,i',j'} + a_{k_2,i,j'} + a_{k_2,i',j} = a_{k_1,i,j'} + a_{k_1,i',j} + a_{k_2,i,j} + a_{k_2,i',j'} \text{ for all distinct slots } k_1,k_2 \in S \\
    \text{ and for all distinct } i,j,i',j' \in V \text{ with } t^\star \notin \{i,i'\} \text{ and with } |k_1-k_2|=1 \text{ if } t^\star \in \{j,j'\}.
    \tag{\S\ref*{thm_flow_out_home_partial_slot}a}
    \label{proof_flow_out_home_partial_slot}
  \end{multline}
  
  For each $k \in S \setminus \{\bar{k}\}$ and each $\ell \in \{5,6,\dots,n\}$ (noting $\ell \neq t^\star = 4$), property~\eqref{proof_flow_out_home_partial_slot} with $(k_1,k_2,i,j,i',j') = (\bar{k},k,1,3,2,
  \ell)$ implies $a_{\bar{k},1,3} + a_{\bar{k},2,\ell} + a_{k,1,\ell} + a_{k,2,3} = a_{\bar{k},1,\ell} + a_{\bar{k},2,3} + a_{k,1,3} + a_{k,2,\ell}$.
  By~\eqref{proof_flow_out_home_basis}, this simplifies to $a_{k,2,\ell} = 0$, from which~\eqref{proof_flow_out_home_home_away} yields $a_{k,\ell,2} = 0$.

  For each $k \in S \setminus \{\bar{k}\}$ and all distinct $\ell,\ell' \in \{3,5,6,\dotsc,n\}$, \eqref{proof_flow_out_home_partial_slot} with $(k_1,k_2,i,j,i',j') = (\bar{k},k,1,\ell',\ell,2)$ implies
  $ a_{\bar{k},1,\ell'} + a_{\bar{k},\ell,2} + a_{k,1,2} + a_{k,\ell,\ell'} = a_{\bar{k},1,2} + a_{\bar{k},\ell,\ell'} + a_{k,1,\ell'} + a_{k,\ell,2}$.
  By~\eqref{proof_flow_out_home_basis} and the previous observation $a_{k,\ell,2} = 0$, this simplifies to $a_{k,\ell,\ell'} = 0$.
  Since also $a_{\bar{k},\star,\star} = \zerovec$, we have
  \begin{equation}
    a_{k,i,j} = 0 \text{ for all } k \in S \text{ and all } i,j \in V \setminus \{t^\star\}.
    \tag{\S\ref*{thm_flow_out_home_partial_slot}b}
    \label{proof_flow_out_home_others}
  \end{equation}
  Let $\ell \in V \setminus \{t^\star\}$.
  For $k \in S \setminus \{\bar{k}\}$, the tuple $(k_1,k_2,i,j,i',j') = (k-1,k,\ell,t^\star,1,2)$ satisfies the conditions of~\eqref{proof_flow_out_home_partial_slot}, and thus for $\ell \in \{3,5,6,\dotsc,n\}$ implies $a_{k-1,\ell,t^\star} + a_{k-1,1,2} + a_{k,\ell,2} + a_{k,1,t^\star} = a_{k-1,\ell,2} + a_{k-1,1,t^\star} + a_{k,\ell,t^\star} + a_{k,1,2}$.
  By~\eqref{proof_flow_out_home_basis} and~\eqref{proof_flow_out_home_others}, this simplifies to $a_{k-1,\ell,t^\star} = a_{k,\ell,t^\star}$.
  By induction on $k$ and $a_{\bar{k},\ell,t^\star} = 0$, we obtain
  \begin{equation}
    a_{k,\ell,t^\star} = 0 \text{ for all } k \in S \text{ and all } \ell \in V \setminus \{t^\star\}.
    \tag{\S\ref*{thm_flow_out_home_partial_slot}c}
    \label{proof_flow_out_home_home_same}
  \end{equation}
  With this, \eqref{proof_flow_out_home_home_away_special} is simplified to
  \begin{equation}
    a_{k,t^\star,j} = a_{k+n-1,t^\star,j} \text{ for all } k \in \{1,2,\dotsc,n-1\} \text{ and all } j \in V \setminus \{t^\star\}.
    \tag{\S\ref*{thm_flow_out_home_partial_slot}d}
    \label{proof_flow_out_home_translated}
  \end{equation}

  \begin{restatable}{claim}{thmFlowOutHomeHomeAwayInterval}
    \label{thm_flow_out_home_home_away_interval}
    For each slot $k \in \{k^\star+1,k^\star+2,\dotsc,k^\star+n-3\}$ and each team $j \in V \setminus \{t^\star\}$ there exist tournaments $T$ and $T'$ satisfying~\eqrefHomeAwaySwap{$k$}{$k+1$}{$j$}{$t^\star$} such that in both tournaments team $t^\star$ leaves its home venue exactly twice and plays away in slots $k^\star$ and $k^\star + n-1$.
  \end{restatable}

  In the tournaments $T$ and $T'$ from \cref{thm_flow_out_home_home_away_interval} team $t^\star$ leaves its home venue exactly twice and does not play home in slots $k^\star$ and $k^\star+n-1$.
  Hence, we have $(\chi(T),\psi(T)), (\chi(T'),\psi(T')) \in F$ and by~\eqref{proof_flow_out_home_coefficients} also $\transpose{b}\psi(T) = \transpose{b}\psi(T')$.
  Combining this with~\eqref{proof_flow_out_home_basis} and~\eqref{proof_flow_out_home_home_same}, equation $\transpose{a}\chi(T) + \transpose{b}\psi(T) = \gamma = \transpose{a}\chi(T') + \transpose{b}\psi(T')$ simplifies to 
  \begin{equation*}
    a_{k+1,t^\star,j} = a_{k,t^\star,j} \text{ for each } k \in S \text{ with } k^\star < k < k^\star+n-1 \text{ and each } j \in V \setminus \{t^\star\}.
  \end{equation*}
  Induction on $k$ yields that $a_{k,t^\star,j}$ is the same for all these $k$.
  Moreover, for each slot $k\in S$ with $k < k^\star$ or $k > k^\star+n-1$ the slot $k+n-1$ (resp.\ $k-n+1$) lies between $k^\star$ and $k^\star+n-1$.
  Application of~\eqref{proof_flow_out_home_translated} yields that $a_{k,t^\star,j}$ is the same for all $k \in S \setminus \{k^\star,k^\star+n-1\}$.
  As $\bar{k}$ is among those, \eqref{proof_flow_out_home_basis} yields
  \begin{equation}
    a_{k,t^\star,j} = 0 \text{ for each } k \in S \setminus \{k^\star,k^\star+n-1\} \text{ and each } j \in V \setminus \{t^\star\}.
    \tag{\S\ref*{thm_flow_out_home_home_away_interval}}
    \label{proof_flow_out_home_home_away_interval}
  \end{equation}
  
  \begin{restatable}{claim}{thmFlowOutHomeHomeAwayMixed}
    \label{thm_flow_out_home_home_away_mixed}
    For all distinct teams $j,j' \in V \setminus \{t^\star\}$ there exist tournaments $T$ and $T'$ satisfying \eqrefHomeAwaySwap{$k^\star$}{$k^\star+1$}{$t^\star$}{$j$} such that team $t^\star$ leaves its home venue exactly once, namely to venue $j$, in tournament $T$ and exactly twice, namely to venues $j$ and $j'$, in tournament $T'$ where it plays away in slots $k^\star$ and $k^\star + n-1$.
  \end{restatable}

  In the tournaments $T$ and $T'$ from \cref{thm_flow_out_home_home_away_mixed} team $t^\star$ leaves its home venue either once or twice, and in the latter case it does not play home in slots $k^\star$ and $k^\star+n-1$.
  Hence, we have $(\chi(T),\psi(T)), (\chi(T'),\psi(T')) \in F$.
  Using~\eqref{proof_flow_out_home_basis}, \eqref{proof_flow_out_home_no_travel}, \eqref{proof_flow_out_home_home_same} and~\eqref{proof_flow_out_home_home_away_interval}, equation $\transpose{a}\chi(T) + \transpose{b}\psi(T) = \gamma = \transpose{a}\chi(T') + \transpose{b}\psi(T')$ simplifies to
  \begin{equation*}
    a_{k^\star,t^\star,j} + b_{t^\star,t^\star,j} = \gamma = b_{t^\star,t^\star,j} + b_{t^\star,t^\star,j'}.
    \tag{\S\ref*{thm_flow_out_home_home_away_mixed}}
  \end{equation*}
  By varying $j$ and $j'$ and considering~\eqref{proof_flow_out_home_translated}, we obtain that $\transpose{a}x + \transpose{b}y \geq \gamma$ is a positive multiple of inequality~\eqref{eq_flow_out_home}.
  This concludes the proof.
\end{proof}

\paragraph{A face defined by flow inequalities.}
Recall the definition of the unconstrained traveling tournament polytope:
\begin{equation*}
  \Ptt(n) \coloneqq \conv\{ (\chi(T),y) \in \{0,1\}^\allMatches \times \{0,1\}^{V \times A} : T \text{ tournament and } y \geq \psi(T) \}.
\end{equation*}
Allowing vectors $y \geq \psi(T)$ augments the set of feasible solutions by suboptimal ones, which is advantageous for finding facet-defining inequalities due to a larger dimension.
Now we examine what happens if we set the flow inequalities~\eqref{eq_flow_out_distinct} and~\eqref{eq_flow_in_distinct} to equality:
\begin{subequations}
  \label{eq_flow_equations}
  \begin{align}
    \sum_{j \in V \setminus \{i\}} y_{t,i,j} &= 1 \quad \forall i,t \in V : i \neq t \label{eq_flow_out_distinct_equations} \\
    \sum_{i \in V \setminus \{j\}} y_{t,i,j} &= 1 \quad \forall j,t \in V : j \neq t \label{eq_flow_in_distinct_equations}
  \end{align}
\end{subequations}
The following theorem shows how we obtain the convex hull of all pairs of play- and travel-vectors as the corresponding face of $\Ptt(n)$.

\begin{theorem}
  \label{thm_face}
  The face of $\Ptt(n)$ defined by equations~\eqref{eq_flow_equations} is equal to
  \begin{equation*}
    \conv\{ (\chi(T),\psi(T)) \in \{0,1\}^\allMatches \times \{0,1\}^{V \times A} : T ~\mathrm{ tournament} \}.
  \end{equation*}
  Consequently, formulation~\eqref{model_basic} together with these equations is an integer programming formulation for this polytope.
\end{theorem}

\begin{proof}
  Let $Q$ be the polytope defined in the statement of the theorem.

  To see that $Q$ is contained in the mentioned face, let $T$ be a tournament.
  For each $i^\star,t^\star \in V$ with $i^\star \neq t^\star$, equation~\eqref{eq_flow_out_distinct_equations} is satisfied by $\psi(T)$ since team $t^\star$ has to play exactly one away match against team $i^\star$ after which it leaves this venue.
  Moreover, it never visits venue $i^\star$ again.
  Similarly, $\psi(T)$ satisfies all equations~\eqref{eq_flow_in_distinct_equations}.
  
  It remains to prove that every vertex $(x,y)$ of the face lies in $Q$.
  Since $\Ptt(n)$ is integral, all its faces are integral as well, and thus $(x,y) \in \{0,1\}^\allMatches \times \{0,1\}^{V \times A}$.
  The vector $x$ defines a tournament $T$ and and we have $y \geq \psi(T)$.
  We have to show $y = \psi(T)$.
  Consider an entry $(t,i,j) \in V \times A$.
  By $i \neq j$, we have $t \neq i$ or $t \neq j$.
  If $t \neq i$, then $y_{t,i,j}$ appears in equation~\eqref{eq_flow_out_distinct_equations} for $(i,t)$ and otherwise it appears in equation~\eqref{eq_flow_in_distinct_equations} for $(j,t)$.
  Since $y$ must be equal to $\psi(T)$ on the support of this equation, we have $y = \psi(T)$, which concludes the proof.
\end{proof}

\section{Problem variants}
\label{sec_problem_variants}

Since $\Ptt(n)$ only reflects the basic constraints for different variants of the traveling tournament problem, we briefly review the variants that occur in the literature.

\paragraph{Mirrored schedules.}
A common requirement is that of mirrored schedules.
Formally, we require that if in slot $k < n$, team $i$ plays home against team $j$, then in slot $k+n-1$, team $i$ plays \emph{away} against team $i$.
Note that the mirroring does not refer to the slots but to the home/away pattern.
This requirement can easily be enforced by adding 
\begin{equation}
  \label{eq_mirror}
  x_{k,i,j} = x_{k+n-1,j,i} \quad \forall k \in \{1,2,\dots,n-1\}, ~\forall (i,j) \in A
\end{equation}
to our model.

\paragraph{No-repeaters.}
In a double round-robin tournament it is often undesirable that the two matches $(k,i,j)$ and $(k',j,i)$ of the teams $i,j \in V$ take place directly after another, i.e., $k$ and $k'$ should not be subsequent numbers.
This can be enforced via the \emph{no-repeater constraints}
\begin{equation}
  \label{eq_no_repeaters}
  x_{k,i,j} + x_{k+1,j,i} \leq 1 \quad \forall k \in \{1,2,\dotsc,2n-3\}, ~\forall (i,j) \in A
\end{equation}

\paragraph{Short home stands and road trips.}
A \emph{home stand} is a set of consecutive matches of team $t \in V$ in which $t$ plays only home.
Similarly, a \emph{road trip} is a set of consecutive matches in which $t$ plays away.
Both such match sequences are undesirable in a tournament, e.g., in order to distribute the home matches of each team more evenly over the season.
For a given parameter $U \in \Z$, the length of home stands and road trips can be restricted to at most $U$ by adding the \emph{home stand} and \emph{road trip constraints}
\begin{subequations}
  \label{eq_homestand_roadtrip}
  \begin{align}
    \sum_{\ell=0}^U \sum_{i \in V \setminus \{t\}} x_{k+\ell,t,i} &\leq U &\quad& \forall k \in \{1,2,\dotsc,2n-2-U\}, ~\forall t \in V, \\
    \sum_{\ell=0}^U \sum_{i \in V \setminus \{t\}} x_{k+\ell,i,t} &\leq U &\quad& \forall k \in \{1,2,\dotsc,2n-2-U\}, ~\forall t \in V,
  \end{align}
\end{subequations}
respectively.
While this is sufficient for the correctness of the model, the requirement has a big effect on the amount of travel of team $t$.
More precisely, the $n-1$ home matches of each team $t$ have to be split into at least $(n-1)/U$ consecutive sequences that are disrupted by away matches.
This implies that team $t$ has to leave (resp.\ enter) its home venue at least this number of times.
The following \emph{home stand flow} and \emph{road trip flow inequalities}
\begin{subequations}
  \label{eq_homestand_roadtrip_flow}
  \begin{align}
    \sum_{j \in V \setminus \{t\}} y_{t,t,j} &\geq \left\lceil \frac{n-1}{U} \right\rceil \quad \forall t \in V \label{eq_homestand_flow} \\
    \sum_{i \in V \setminus \{t\}} y_{t,i,t} &\geq \left\lceil \frac{n-1}{U} \right\rceil \quad \forall t \in V \label{eq_roadtrip_flow}
  \end{align}
\end{subequations}
model this effect.
The fact that these inequalities actually strengthen the LP relaxation (and are thus not just implied by~\eqref{eq_homestand_roadtrip}) will become clear in the next section.

\section{Computational impact}
\label{sec_computations}

\DeclareDocumentCommand\mirrored{}{\Yinyang}
\DeclareDocumentCommand\eqFlowDistinctSame{}{\eqref{eq_flow_distinct}$_{i=t}$}

\begin{table}[htpb]
  \caption{%
    Sources of considered test instances.
    Instance files were obtained from \cite{ROBINXREPO}. 
  }%
  \label{tab_instance_sources}
  \begin{center}
    {\small{%
    \setlength{\tabcolsep}{2pt}
    \renewcommand{\arraystretch}{1.03}
    \begin{tabular}{l|l|l}
      \textbf{Class} & \textbf{Source} & \textbf{Description} \\ \hline
      NL$\langle n \rangle$ & \cite{EastonNT01} & Air distances of cities in National League of Major League Baseball \\
      SUP$\langle n \rangle$ & \cite{UthusRG09,UthusRG12} & Air distances of cities in Super 14 rugby cup \\
      GAL$\langle n \rangle$ & \cite{UthusRG09,UthusRG12} & Venues are exoplanet locations in 3D-space \\
      INCR $\langle n \rangle$ & \cite{HoshinoK12} & Venues are on straight line, increasing distance \\
      LINE$\langle n \rangle$ & \cite{HoshinoK12} & Venues are equidistant on straight line \\
      CIRC$\langle n \rangle$ & \cite{EastonNT01} & Venues are equidistant on circle \\
      CON$\langle n \rangle$ & \cite{UrrutiaR06} & Distance is constant
    \end{tabular}
    }}%
  \end{center}
\end{table}

\begin{table}[htpb]
  \caption{%
    Characteristics of considered test instances.
    Number of variables, constraints and nonzeros reflect these numbers after Gurobi's presolve.
  }
  \label{tab_instance_characteristics}
  \begin{center}
    {\small{%
    \setlength{\tabcolsep}{2pt}
    \renewcommand{\arraystretch}{1.03}
    \begin{tabular}{l|l|rrrr|rrrr}
      \textbf{Variant} & \textbf{Teams} & \multicolumn{4}{c|}{\textbf{Plain}} & \multicolumn{4}{c}{\textbf{Mirrored \mirrored}} \\ \hline
      \textbf{Base} & & \eqref{eq_no_repeaters}--\eqref{eq_homestand_roadtrip} & \eqref{eq_travel_away_away_lifted}--\eqref{eq_travel_away_home_lifted} & \eqref{eq_flow_distinct} & \eqref{eq_homestand_roadtrip_flow} & \eqref{eq_mirror}--\eqref{eq_homestand_roadtrip} & \eqref{eq_travel_away_away_lifted}--\eqref{eq_travel_away_home_lifted} & \eqref{eq_flow_distinct} & \eqref{eq_homestand_roadtrip_flow} \\ \hline
      \textbf{Parameter $U$} & $\{4,6,8\}$ & 3 & & &  & 3 & & & \\ \hline
      \textbf{Variables} & $4$ & 120 &&& & 84 &&&  \\
      \textbf{+ additional} & $6$ & 480 &&& & 330 &&&  \\
      & $8$ & 1232 &&& & 840 &&& \\ \hline
      \textbf{Constraints} & $4$ & 332 & +120 & +24 & +8 & 282 & +120 & +24 & +8 \\
      \textbf{+ additional}& $6$ & 1998 & +1080 & +60 & +12 & 1785 & +1134 & +60 & +12 \\
      & $8$ & 6664 & +5268 & +112 & +16 & 6132 & +7655 & +112 & +16 \\ \hline
      \textbf{Nonzeros} & $4$ & 1334 & +600 & +72 & +24 & 1116 & +600 & +72 & +24 \\
      & $6$ & 10260 & +5400 & +300 & +60 & 8790 & +5562 & +300 & +60 \\
      & $8$ & 35168 & +24540 & +784 & +112 & 30744 & +31701 & +784 & +112
    \end{tabular}
    }}%
  \end{center}
\end{table}

In this section we evaluate the addition of the inequalities that were discussed theoretically in a practical setting.
To this end, we implemented the IP models in Gurobi~9.5~\cite{GUROBI95} and assessed the impact for various instances\footnote{Our implementation can be obtained from github: \href{https://github.com/discopt/traveling-tournament-cubic}{github.com/discopt/traveling-tournament-cubic}.}.
We ran our experiments on an Intel Xeon Gold 5217 CPU with \SI{3.00}{\giga\hertz} with \SI{64}{\giga\byte} memory, on a single thread and with a time limit of 1\,hour.
Our testbed consists of instances that were used previously by~\cite{EastonNT01,RibeiroU07,UrrutiaR06,UthusRG09,UthusRG12}.
We made use of the RobinX instance repository~\cite{ROBINXREPO} and unified instance format~\cite{ROBINXFORMAT}.
The sources and characteristics of the instances are depicted in \cref{tab_instance_sources} and \cref{tab_instance_characteristics}, respectively.
We aggregated the counts for constraints and nonzeros for \eqref{eq_travel_away_away_lifted}--\eqref{eq_travel_away_home_lifted} since these all constitute lifted model inequalities.
It is easy to see that are actually quite many such inequalities, and their number is dominated by~\eqref{eq_travel_away_away_lifted} of which there exist $\orderTheta(n^4)$ many.

As can be seen from the tables, the integer programs neither have many variables nor many constraints.
Moreover, the instance are not particularly dense.
Nevertheless, the instances for $n=6$ (plain) and $n=8$ (plain and mirrored) are already hard to solve, which is why we do not report about computational results for larger problem sizes.

\begin{table}[htpb]
  \caption{%
    Quality of the LP bounds after adding different (sets of) constraints to the base model (see \cref{tab_instance_characteristics}) or after removal of (sets of) constraints from the full model, which is the base model augmented by~\eqref{eq_travel_away_away_lifted}-\eqref{eq_travel_away_home_lifted}, \eqref{eq_flow_home}, \eqref{eq_flow_equations} and \eqref{eq_homestand_roadtrip_flow}.
    The percentages indicate the ratio ``LP bound''/``best known primal solution''.
    \eqFlowDistinctSame{} indicates constraints~\eqref{eq_flow_distinct} for $i=t$.
    Mirrored instances are indicated via \mirrored{}.
  }%
  \label{tab_lp_bounds}
  \begin{center}
    {\small{%
    \setlength{\tabcolsep}{2pt}
    \renewcommand{\arraystretch}{1.03}
    \begin{tabular}{l|rrrrr|rrrr}
      \textbf{Instance} & \multicolumn{5}{c|}{\textbf{Addition of constraints}} & \multicolumn{4}{c}{\textbf{Removal of constraints}} \\
      & \textbf{Base} & \eqref{eq_travel_away_away_lifted}--\eqref{eq_travel_away_home_lifted} & \eqref{eq_flow_distinct},\eqFlowDistinctSame{} & \eqref{eq_flow_home},\eqref{eq_flow_equations} & \eqref{eq_homestand_roadtrip_flow} & \textbf{Full} & \eqref{eq_travel_away_away_lifted}--\eqref{eq_travel_away_home_lifted} & \eqref{eq_flow_home},\eqref{eq_flow_equations} & \eqref{eq_homestand_roadtrip_flow} \\ \hline
NL4 \mirrored{} & \SI{24.3}{\percent} & \SI{24.3}{\percent} & \SI{97.0}{\percent} & \SI{97.0}{\percent} & \SI{30.8}{\percent} & \SI{97.0}{\percent} & \SI{97.0}{\percent} & \SI{34.6}{\percent} & \SI{97.0}{\percent} \\
SUP4 \mirrored{} & \SI{24.9}{\percent} & \SI{24.9}{\percent} & \SI{41.0}{\percent} & \SI{41.0}{\percent} & \SI{28.3}{\percent} & \SI{41.0}{\percent} & \SI{41.0}{\percent} & \SI{28.3}{\percent} & \SI{41.0}{\percent} \\
GAL4 \mirrored{} & \SI{24.8}{\percent} & \SI{24.8}{\percent} & \SI{94.1}{\percent} & \SI{94.1}{\percent} & \SI{35.4}{\percent} & \SI{94.1}{\percent} & \SI{94.1}{\percent} & \SI{38.3}{\percent} & \SI{94.1}{\percent} \\
INCR4 \mirrored{} & \SI{25.0}{\percent} & \SI{25.0}{\percent} & \SI{77.1}{\percent} & \SI{77.1}{\percent} & \SI{35.4}{\percent} & \SI{77.1}{\percent} & \SI{77.1}{\percent} & \SI{37.5}{\percent} & \SI{77.1}{\percent} \\
LINE4 \mirrored{} & \SI{25.0}{\percent} & \SI{25.0}{\percent} & \SI{77.8}{\percent} & \SI{77.8}{\percent} & \SI{41.7}{\percent} & \SI{77.8}{\percent} & \SI{77.8}{\percent} & \SI{41.7}{\percent} & \SI{77.8}{\percent} \\
CIRC4 \mirrored{} & \SI{20.0}{\percent} & \SI{20.0}{\percent} & \SI{80.0}{\percent} & \SI{80.0}{\percent} & \SI{40.0}{\percent} & \SI{80.0}{\percent} & \SI{80.0}{\percent} & \SI{40.0}{\percent} & \SI{80.0}{\percent} \\
CON4 \mirrored{} & \SI{23.5}{\percent} & \SI{23.5}{\percent} & \SI{94.1}{\percent} & \SI{94.1}{\percent} & \SI{47.1}{\percent} & \SI{94.1}{\percent} & \SI{94.1}{\percent} & \SI{47.1}{\percent} & \SI{94.1}{\percent} \\
NL4 & \SI{24.2}{\percent} & \SI{24.2}{\percent} & \SI{96.9}{\percent} & \SI{96.9}{\percent} & \SI{30.4}{\percent} & \SI{96.9}{\percent} & \SI{96.9}{\percent} & \SI{32.6}{\percent} & \SI{96.9}{\percent} \\
SUP4 & \SI{5.2}{\percent} & \SI{5.2}{\percent} & \SI{20.9}{\percent} & \SI{20.9}{\percent} & \SI{10.4}{\percent} & \SI{20.9}{\percent} & \SI{20.9}{\percent} & \SI{10.4}{\percent} & \SI{20.9}{\percent} \\
GAL4 & \SI{22.6}{\percent} & \SI{22.6}{\percent} & \SI{90.4}{\percent} & \SI{90.4}{\percent} & \SI{35.0}{\percent} & \SI{90.4}{\percent} & \SI{90.4}{\percent} & \SI{36.7}{\percent} & \SI{90.4}{\percent} \\
INCR4 & \SI{16.7}{\percent} & \SI{16.7}{\percent} & \SI{66.7}{\percent} & \SI{66.7}{\percent} & \SI{31.3}{\percent} & \SI{66.7}{\percent} & \SI{66.7}{\percent} & \SI{31.3}{\percent} & \SI{66.7}{\percent} \\
LINE4 & \SI{16.7}{\percent} & \SI{16.7}{\percent} & \SI{66.7}{\percent} & \SI{66.7}{\percent} & \SI{33.3}{\percent} & \SI{66.7}{\percent} & \SI{66.7}{\percent} & \SI{33.3}{\percent} & \SI{66.7}{\percent} \\
CIRC4 & \SI{20.0}{\percent} & \SI{20.0}{\percent} & \SI{80.0}{\percent} & \SI{80.0}{\percent} & \SI{40.0}{\percent} & \SI{80.0}{\percent} & \SI{80.0}{\percent} & \SI{40.0}{\percent} & \SI{80.0}{\percent} \\
CON4 & \SI{23.5}{\percent} & \SI{23.5}{\percent} & \SI{94.1}{\percent} & \SI{94.1}{\percent} & \SI{47.1}{\percent} & \SI{94.1}{\percent} & \SI{94.1}{\percent} & \SI{47.1}{\percent} & \SI{94.1}{\percent} \\
\hline
NL6 \mirrored{} & \SI{11.0}{\percent} & \SI{11.0}{\percent} & \SI{53.2}{\percent} & \SI{53.2}{\percent} & \SI{30.1}{\percent} & \SI{65.5}{\percent} & \SI{65.5}{\percent} & \SI{30.7}{\percent} & \SI{53.2}{\percent} \\
SUP6 \mirrored{} & \SI{10.8}{\percent} & \SI{10.8}{\percent} & \SI{14.1}{\percent} & \SI{14.1}{\percent} & \SI{12.6}{\percent} & \SI{29.1}{\percent} & \SI{29.1}{\percent} & \SI{12.6}{\percent} & \SI{14.1}{\percent} \\
GAL6 \mirrored{} & \SI{11.3}{\percent} & \SI{11.3}{\percent} & \SI{65.1}{\percent} & \SI{65.1}{\percent} & \SI{35.6}{\percent} & \SI{77.2}{\percent} & \SI{77.2}{\percent} & \SI{36.1}{\percent} & \SI{65.1}{\percent} \\
INCR6 \mirrored{} & \SI{9.0}{\percent} & \SI{9.0}{\percent} & \SI{44.2}{\percent} & \SI{44.2}{\percent} & \SI{26.0}{\percent} & \SI{56.7}{\percent} & \SI{56.7}{\percent} & \SI{26.7}{\percent} & \SI{45.0}{\percent} \\
LINE6 \mirrored{} & \SI{8.9}{\percent} & \SI{8.9}{\percent} & \SI{44.6}{\percent} & \SI{44.6}{\percent} & \SI{28.9}{\percent} & \SI{57.8}{\percent} & \SI{57.8}{\percent} & \SI{28.9}{\percent} & \SI{45.2}{\percent} \\
CIRC6 \mirrored{} & \SI{8.3}{\percent} & \SI{8.3}{\percent} & \SI{50.0}{\percent} & \SI{50.0}{\percent} & \SI{33.3}{\percent} & \SI{66.7}{\percent} & \SI{66.7}{\percent} & \SI{33.3}{\percent} & \SI{50.0}{\percent} \\
CON6 \mirrored{} & \SI{12.5}{\percent} & \SI{12.5}{\percent} & \SI{75.0}{\percent} & \SI{75.0}{\percent} & \SI{50.0}{\percent} & \SI{87.5}{\percent} & \SI{87.5}{\percent} & \SI{50.0}{\percent} & \SI{75.0}{\percent} \\
NL6 & \SI{9.1}{\percent} & \SI{9.1}{\percent} & \SI{54.8}{\percent} & \SI{54.8}{\percent} & \SI{32.1}{\percent} & \SI{72.8}{\percent} & \SI{72.8}{\percent} & \SI{32.1}{\percent} & \SI{54.8}{\percent} \\
SUP6 & \SI{0.7}{\percent} & \SI{0.7}{\percent} & \SI{4.2}{\percent} & \SI{4.2}{\percent} & \SI{2.8}{\percent} & \SI{32.9}{\percent} & \SI{32.9}{\percent} & \SI{2.8}{\percent} & \SI{4.2}{\percent} \\
GAL6 & \SI{12.0}{\percent} & \SI{12.0}{\percent} & \SI{72.1}{\percent} & \SI{72.1}{\percent} & \SI{39.8}{\percent} & \SI{87.3}{\percent} & \SI{87.3}{\percent} & \SI{39.8}{\percent} & \SI{72.1}{\percent} \\
INCR6 & \SI{7.9}{\percent} & \SI{7.9}{\percent} & \SI{47.4}{\percent} & \SI{47.4}{\percent} & \SI{28.9}{\percent} & \SI{66.7}{\percent} & \SI{66.7}{\percent} & \SI{28.9}{\percent} & \SI{47.4}{\percent} \\
LINE6 & \SI{7.9}{\percent} & \SI{7.9}{\percent} & \SI{47.4}{\percent} & \SI{47.4}{\percent} & \SI{31.6}{\percent} & \SI{68.4}{\percent} & \SI{68.4}{\percent} & \SI{31.6}{\percent} & \SI{47.4}{\percent} \\
CIRC6 & \SI{9.4}{\percent} & \SI{9.4}{\percent} & \SI{56.3}{\percent} & \SI{56.3}{\percent} & \SI{37.5}{\percent} & \SI{75.0}{\percent} & \SI{75.0}{\percent} & \SI{37.5}{\percent} & \SI{56.3}{\percent} \\
CON6 & \SI{14.0}{\percent} & \SI{14.0}{\percent} & \SI{83.7}{\percent} & \SI{83.7}{\percent} & \SI{55.8}{\percent} & \SI{97.7}{\percent} & \SI{97.7}{\percent} & \SI{55.8}{\percent} & \SI{83.7}{\percent} \\
\hline
NL8 \mirrored{} & \SI{8.2}{\percent} & \SI{8.2}{\percent} & \SI{53.8}{\percent} & \SI{53.8}{\percent} & \SI{33.3}{\percent} & \SI{76.1}{\percent} & \SI{76.1}{\percent} & \SI{33.8}{\percent} & \SI{53.8}{\percent} \\
SUP8 \mirrored{} & \SI{1.7}{\percent} & \SI{1.7}{\percent} & \SI{7.2}{\percent} & \SI{7.2}{\percent} & \SI{4.1}{\percent} & \SI{32.1}{\percent} & \SI{32.1}{\percent} & \SI{4.1}{\percent} & \SI{7.2}{\percent} \\
GAL8 \mirrored{} & \SI{7.9}{\percent} & \SI{7.9}{\percent} & \SI{49.6}{\percent} & \SI{49.6}{\percent} & \SI{31.4}{\percent} & \SI{71.9}{\percent} & \SI{71.9}{\percent} & \SI{31.5}{\percent} & \SI{49.6}{\percent} \\
INCR8 \mirrored{} & \SI{6.0}{\percent} & \SI{6.0}{\percent} & \SI{37.2}{\percent} & \SI{37.2}{\percent} & \SI{25.0}{\percent} & \SI{56.4}{\percent} & \SI{56.4}{\percent} & \SI{25.3}{\percent} & \SI{37.3}{\percent} \\
LINE8 \mirrored{} & \SI{6.0}{\percent} & \SI{6.0}{\percent} & \SI{37.7}{\percent} & \SI{37.7}{\percent} & \SI{27.7}{\percent} & \SI{56.5}{\percent} & \SI{56.5}{\percent} & \SI{27.7}{\percent} & \SI{37.7}{\percent} \\
CIRC8 \mirrored{} & \SI{5.7}{\percent} & \SI{5.7}{\percent} & \SI{45.7}{\percent} & \SI{45.7}{\percent} & \SI{34.3}{\percent} & \SI{68.6}{\percent} & \SI{68.6}{\percent} & \SI{34.3}{\percent} & \SI{45.7}{\percent} \\
CON8 \mirrored{} & \SI{10.0}{\percent} & \SI{10.0}{\percent} & \SI{80.0}{\percent} & \SI{80.0}{\percent} & \SI{60.0}{\percent} & \SI{100.0}{\percent} & \SI{100.0}{\percent} & \SI{60.0}{\percent} & \SI{80.0}{\percent} \\
NL8 & \SI{6.8}{\percent} & \SI{6.8}{\percent} & \SI{54.1}{\percent} & \SI{54.1}{\percent} & \SI{34.4}{\percent} & \SI{80.4}{\percent} & \SI{80.4}{\percent} & \SI{34.4}{\percent} & \SI{54.1}{\percent} \\
SUP8 & \SI{1.0}{\percent} & \SI{1.0}{\percent} & \SI{7.7}{\percent} & \SI{7.7}{\percent} & \SI{3.9}{\percent} & \SI{39.8}{\percent} & \SI{39.8}{\percent} & \SI{3.9}{\percent} & \SI{7.7}{\percent} \\
GAL8 & \SI{6.5}{\percent} & \SI{6.5}{\percent} & \SI{52.3}{\percent} & \SI{52.3}{\percent} & \SI{32.5}{\percent} & \SI{78.8}{\percent} & \SI{78.8}{\percent} & \SI{32.5}{\percent} & \SI{52.3}{\percent} \\
INCR8 & \SI{5.1}{\percent} & \SI{5.1}{\percent} & \SI{41.0}{\percent} & \SI{41.0}{\percent} & \SI{28.4}{\percent} & \SI{66.7}{\percent} & \SI{66.7}{\percent} & \SI{28.4}{\percent} & \SI{41.0}{\percent} \\
LINE8 & \SI{4.9}{\percent} & \SI{4.9}{\percent} & \SI{39.5}{\percent} & \SI{39.5}{\percent} & \SI{29.6}{\percent} & \SI{64.2}{\percent} & \SI{64.2}{\percent} & \SI{29.6}{\percent} & \SI{39.5}{\percent} \\
CIRC8 & \SI{6.1}{\percent} & \SI{6.1}{\percent} & \SI{48.5}{\percent} & \SI{48.5}{\percent} & \SI{36.4}{\percent} & \SI{72.7}{\percent} & \SI{72.7}{\percent} & \SI{36.4}{\percent} & \SI{48.5}{\percent} \\
CON8 & \SI{10.0}{\percent} & \SI{10.0}{\percent} & \SI{80.0}{\percent} & \SI{80.0}{\percent} & \SI{60.0}{\percent} & \SI{100.0}{\percent} & \SI{100.0}{\percent} & \SI{60.0}{\percent} & \SI{80.0}{\percent} \\
    \end{tabular}
    }}%
  \end{center}
\end{table}

\paragraph{LP bounds.}
Let us consider the quality of the LP relaxations that we obtain after adding the additional inequalities.
\cref{tab_lp_bounds} depicts these values relative to the corresponding best known solution value.
The latter were taken from the the RobinX instance repository~\cite{ROBINXREPO}.
While the left part of the table is about the improvement after adding a certain class of inequalities, the right part shows what happens if we remove such a class from the strongest possible model, which is the base model augmented by~\eqref{eq_travel_away_away_lifted}--\eqref{eq_travel_away_home_lifted}, \eqref{eq_flow_home}, \eqref{eq_flow_equations} and \eqref{eq_homestand_roadtrip_flow}.
Note that~\eqref{eq_flow_distinct} is not considered in this full model since these inequalities are implied by~\eqref{eq_flow_equations} for $i \neq t$ and by~\eqref{eq_flow_home} for $i = t$.

First, note that adding or removing the lifted versions \eqref{eq_travel_away_away_lifted}--\eqref{eq_travel_away_home_lifted} of the model inequalites does not affect the LP bounds.
This matches the theoretical observation that the dimensions of the respective faces are already quite high, i.e., the model inequalities are almost facet defining.
Similarly, it does not matter whether we add the flow inequalities~\eqref{eq_flow_distinct} or the flow equations~\eqref{eq_flow_equations} for $i \neq t$ or whether we add the flow inequalities~\eqref{eq_flow_distinct} for $i=t$ or their strengthened version, the home-flow inequalities~\eqref{eq_flow_home}.
Flow inequalites themselves clearly have the biggest impact, but also the home stand flow and road trip flow inequalities~\eqref{eq_homestand_roadtrip_flow} are quite useful.
Together, they already provide the best LP bounds that we can obtain with all our proposed inequalities.

\begin{table}[htpb]
  \caption{%
    Quality of the IP bounds after adding different (sets of) constraints to the base model (see \cref{tab_instance_characteristics}) or after removal of (sets of) constraints from the full model, which is the base model augmented by~\eqref{eq_travel_away_away_lifted}-\eqref{eq_travel_away_home_lifted}, \eqref{eq_flow_home}, \eqref{eq_flow_equations} and \eqref{eq_homestand_roadtrip_flow}.
    The percentages indicate the ratio ``IP bound''/``best known primal solution''.
    \eqFlowDistinctSame{} indicates constraints~\eqref{eq_flow_distinct} for $i=t$.
    Mirrored instances are indicated via \mirrored{}.
  }%
  \label{tab_ip_bounds}
  \begin{center}
    {\small{%
    \setlength{\tabcolsep}{2pt}
    \renewcommand{\arraystretch}{1.03}
    \begin{tabular}{l|rrrrr|rrrr}
      \textbf{Instance} & \multicolumn{5}{c|}{\textbf{Addition of constraints}} & \multicolumn{4}{c}{\textbf{Removal of constraints}} \\
      & \textbf{Base} & \eqref{eq_travel_away_away_lifted}--\eqref{eq_travel_away_home_lifted} & \eqref{eq_flow_distinct},\eqFlowDistinctSame{} & \eqref{eq_flow_home},\eqref{eq_flow_equations} & \eqref{eq_homestand_roadtrip_flow} & \textbf{Full} & \eqref{eq_travel_away_away_lifted}--\eqref{eq_travel_away_home_lifted} & \eqref{eq_flow_home},\eqref{eq_flow_equations} & \eqref{eq_homestand_roadtrip_flow} \\ \hline
NL6 \mirrored{} & \SI{100.0}{\percent} & \SI{100.0}{\percent} & \SI{100.0}{\percent} & \SI{100.0}{\percent} & \SI{100.0}{\percent} & \SI{100.0}{\percent} & \SI{100.0}{\percent} & \SI{100.0}{\percent} & \SI{100.0}{\percent} \\
SUP6 \mirrored{} & \SI{100.0}{\percent} & \SI{100.0}{\percent} & \SI{100.0}{\percent} & \SI{100.0}{\percent} & \SI{100.0}{\percent} & \SI{100.0}{\percent} & \SI{100.0}{\percent} & \SI{100.0}{\percent} & \SI{100.0}{\percent} \\
GAL6 \mirrored{} & \SI{100.0}{\percent} & \SI{100.0}{\percent} & \SI{100.0}{\percent} & \SI{100.0}{\percent} & \SI{75.9}{\percent} & \SI{100.0}{\percent} & \SI{95.0}{\percent} & \SI{85.3}{\percent} & \SI{100.0}{\percent} \\
INCR6 \mirrored{} & \SI{100.0}{\percent} & \SI{100.0}{\percent} & \SI{100.0}{\percent} & \SI{100.0}{\percent} & \SI{100.0}{\percent} & \SI{100.0}{\percent} & \SI{100.0}{\percent} & \SI{100.0}{\percent} & \SI{100.0}{\percent} \\
LINE6 \mirrored{} & \SI{100.0}{\percent} & \SI{100.0}{\percent} & \SI{100.0}{\percent} & \SI{100.0}{\percent} & \SI{100.0}{\percent} & \SI{100.0}{\percent} & \SI{100.0}{\percent} & \SI{100.0}{\percent} & \SI{100.0}{\percent} \\
CIRC6 \mirrored{} & \SI{100.0}{\percent} & \SI{100.0}{\percent} & \SI{100.0}{\percent} & \SI{100.0}{\percent} & \SI{100.0}{\percent} & \SI{100.0}{\percent} & \SI{100.0}{\percent} & \SI{100.0}{\percent} & \SI{100.0}{\percent} \\
CON6 \mirrored{} & \SI{100.0}{\percent} & \SI{100.0}{\percent} & \SI{100.0}{\percent} & \SI{100.0}{\percent} & \SI{100.0}{\percent} & \SI{100.0}{\percent} & \SI{100.0}{\percent} & \SI{100.0}{\percent} & \SI{100.0}{\percent} \\
NL6 & \SI{42.7}{\percent} & \SI{40.3}{\percent} & \SI{70.3}{\percent} & \SI{69.8}{\percent} & \SI{63.9}{\percent} & \SI{83.6}{\percent} & \SI{84.1}{\percent} & \SI{57.9}{\percent} & \SI{69.0}{\percent} \\
SUP6 & \SI{35.2}{\percent} & \SI{36.5}{\percent} & \SI{37.7}{\percent} & \SI{45.6}{\percent} & \SI{45.0}{\percent} & \SI{42.4}{\percent} & \SI{47.8}{\percent} & \SI{46.7}{\percent} & \SI{41.3}{\percent} \\
GAL6 & \SI{44.2}{\percent} & \SI{39.9}{\percent} & \SI{76.3}{\percent} & \SI{77.4}{\percent} & \SI{66.1}{\percent} & \SI{89.0}{\percent} & \SI{88.7}{\percent} & \SI{65.6}{\percent} & \SI{78.3}{\percent} \\
INCR6 & \SI{45.2}{\percent} & \SI{42.1}{\percent} & \SI{66.2}{\percent} & \SI{65.8}{\percent} & \SI{59.2}{\percent} & \SI{70.6}{\percent} & \SI{72.8}{\percent} & \SI{62.3}{\percent} & \SI{66.2}{\percent} \\
LINE6 & \SI{53.9}{\percent} & \SI{46.1}{\percent} & \SI{67.1}{\percent} & \SI{67.1}{\percent} & \SI{57.9}{\percent} & \SI{73.7}{\percent} & \SI{76.3}{\percent} & \SI{59.2}{\percent} & \SI{68.4}{\percent} \\
CIRC6 & \SI{62.5}{\percent} & \SI{40.6}{\percent} & \SI{64.1}{\percent} & \SI{75.0}{\percent} & \SI{75.0}{\percent} & \SI{78.1}{\percent} & \SI{76.6}{\percent} & \SI{59.4}{\percent} & \SI{75.0}{\percent} \\
CON6 & \SI{88.4}{\percent} & \SI{83.7}{\percent} & \SI{88.4}{\percent} & \SI{83.7}{\percent} & \SI{100.0}{\percent} & \SI{97.7}{\percent} & \SI{97.7}{\percent} & \SI{74.4}{\percent} & \SI{86.0}{\percent} \\
\hline
NL8 \mirrored{} & \SI{29.5}{\percent} & \SI{19.8}{\percent} & \SI{58.8}{\percent} & \SI{59.4}{\percent} & \SI{61.3}{\percent} & \SI{77.9}{\percent} & \SI{81.3}{\percent} & \SI{57.8}{\percent} & \SI{60.9}{\percent} \\
SUP8 \mirrored{} & \SI{11.6}{\percent} & \SI{10.9}{\percent} & \SI{18.3}{\percent} & \SI{20.3}{\percent} & \SI{32.1}{\percent} & \SI{32.1}{\percent} & \SI{37.4}{\percent} & \SI{34.1}{\percent} & \SI{15.5}{\percent} \\
GAL8 \mirrored{} & \SI{27.3}{\percent} & \SI{19.2}{\percent} & \SI{54.2}{\percent} & \SI{54.8}{\percent} & \SI{59.2}{\percent} & \SI{76.9}{\percent} & \SI{76.8}{\percent} & \SI{57.0}{\percent} & \SI{52.9}{\percent} \\
INCR8 \mirrored{} & \SI{22.2}{\percent} & \SI{18.3}{\percent} & \SI{44.6}{\percent} & \SI{44.6}{\percent} & \SI{47.6}{\percent} & \SI{63.8}{\percent} & \SI{61.1}{\percent} & \SI{45.8}{\percent} & \SI{43.1}{\percent} \\
LINE8 \mirrored{} & \SI{22.3}{\percent} & \SI{18.5}{\percent} & \SI{44.0}{\percent} & \SI{46.7}{\percent} & \SI{47.3}{\percent} & \SI{60.9}{\percent} & \SI{62.0}{\percent} & \SI{44.6}{\percent} & \SI{42.9}{\percent} \\
CIRC8 \mirrored{} & \SI{46.4}{\percent} & \SI{30.7}{\percent} & \SI{51.4}{\percent} & \SI{51.4}{\percent} & \SI{62.9}{\percent} & \SI{70.7}{\percent} & \SI{70.7}{\percent} & \SI{50.7}{\percent} & \SI{51.4}{\percent} \\
CON8 \mirrored{} & \SI{57.5}{\percent} & \SI{37.5}{\percent} & \SI{86.3}{\percent} & \SI{81.3}{\percent} & \SI{100.0}{\percent} & \SI{100.0}{\percent} & \SI{100.0}{\percent} & \SI{85.0}{\percent} & \SI{81.3}{\percent} \\
NL8 & \SI{18.0}{\percent} & \SI{11.9}{\percent} & \SI{56.2}{\percent} & \SI{55.9}{\percent} & \SI{60.2}{\percent} & \SI{80.4}{\percent} & \SI{80.4}{\percent} & \SI{58.4}{\percent} & \SI{55.6}{\percent} \\
SUP8 & \SI{9.2}{\percent} & \SI{8.8}{\percent} & \SI{15.4}{\percent} & \SI{15.3}{\percent} & \SI{35.9}{\percent} & \SI{39.8}{\percent} & \SI{39.8}{\percent} & \SI{35.9}{\percent} & \SI{15.2}{\percent} \\
GAL8 & \SI{15.8}{\percent} & \SI{12.6}{\percent} & \SI{55.0}{\percent} & \SI{55.7}{\percent} & \SI{55.2}{\percent} & \SI{78.8}{\percent} & \SI{78.8}{\percent} & \SI{55.2}{\percent} & \SI{54.4}{\percent} \\
INCR8 & \SI{15.1}{\percent} & \SI{11.1}{\percent} & \SI{43.4}{\percent} & \SI{44.9}{\percent} & \SI{43.6}{\percent} & \SI{66.7}{\percent} & \SI{66.7}{\percent} & \SI{43.6}{\percent} & \SI{44.6}{\percent} \\
LINE8 & \SI{26.5}{\percent} & \SI{14.2}{\percent} & \SI{42.6}{\percent} & \SI{44.4}{\percent} & \SI{63.0}{\percent} & \SI{64.2}{\percent} & \SI{64.2}{\percent} & \SI{44.4}{\percent} & \SI{42.6}{\percent} \\
CIRC8 & \SI{39.4}{\percent} & \SI{24.2}{\percent} & \SI{48.5}{\percent} & \SI{49.2}{\percent} & \SI{69.7}{\percent} & \SI{72.7}{\percent} & \SI{72.7}{\percent} & \SI{57.6}{\percent} & \SI{48.5}{\percent} \\
CON8 & \SI{80.0}{\percent} & \SI{32.5}{\percent} & \SI{80.0}{\percent} & \SI{80.0}{\percent} & \SI{95.0}{\percent} & \SI{100.0}{\percent} & \SI{100.0}{\percent} & \SI{62.5}{\percent} & \SI{80.0}{\percent} \\
    \end{tabular}
    }}%
  \end{center}
\end{table}

\begin{table}[htpb]
  \caption{%
    Number of branch-and-bound nodes after adding different (sets of) constraints to the base model (see \cref{tab_instance_characteristics}) or after removal of (sets of) constraints from the full model, which is the base model augmented by~\eqref{eq_travel_away_away_lifted}-\eqref{eq_travel_away_home_lifted}, \eqref{eq_flow_home}, \eqref{eq_flow_equations} and \eqref{eq_homestand_roadtrip_flow}.
    Node numbers are given as multiples of $1000$.
    \eqFlowDistinctSame{} indicates constraints~\eqref{eq_flow_distinct} for $i=t$.
    Mirrored instances are indicated via \mirrored{}.
  }%
  \label{tab_ip_nodes}
  \begin{center}
    {\small{%
    \setlength{\tabcolsep}{2pt}
    \renewcommand{\arraystretch}{1.03}
    \begin{tabular}{l|rrrrr|rrrr}
      \textbf{Instance} & \multicolumn{5}{c|}{\textbf{Addition of constraints}} & \multicolumn{4}{c}{\textbf{Removal of constraints}} \\
      & \textbf{Base} & \eqref{eq_travel_away_away_lifted}--\eqref{eq_travel_away_home_lifted} & \eqref{eq_flow_distinct},\eqFlowDistinctSame{} & \eqref{eq_flow_home},\eqref{eq_flow_equations} & \eqref{eq_homestand_roadtrip_flow} & \textbf{Full} & \eqref{eq_travel_away_away_lifted}--\eqref{eq_travel_away_home_lifted} & \eqref{eq_flow_home},\eqref{eq_flow_equations} & \eqref{eq_homestand_roadtrip_flow} \\ \hline
NL6 \mirrored{} & 202 & 161 & 190 & 69 & 345 & 39 & 57 & 169 & 53 \\
SUP6 \mirrored{} & 78 & 96 & 117 & 41 & 74 & 35 & 12 & 79 & 40 \\
GAL6 \mirrored{} & 248 & 275 & 343 & 282 & 270 & 149 & 301 & 520 & 196 \\
INCR6 \mirrored{} & 140 & 137 & 80 & 90 & 131 & 17 & 58 & 120 & 56 \\
LINE6 \mirrored{} & 96 & 183 & 86 & 53 & 122 & 70 & 19 & 87 & 11 \\
CIRC6 \mirrored{} & 80 & 119 & 18 & 9 & 354 & 87 & 26 & 234 & 9 \\
CON6 \mirrored{} & 8 & 12 & 2 & 46 & 3 & 208 & 104 & 93 & 30 \\
NL6 & 54 & 51 & 101 & 81 & 105 & 100 & 145 & 158 & 104 \\
SUP6 & 64 & 54 & 87 & 72 & 301 & 184 & 303 & 91 & 54 \\
GAL6 & 56 & 49 & 103 & 63 & 94 & 169 & 274 & 94 & 71 \\
INCR6 & 71 & 65 & 125 & 78 & 86 & 109 & 270 & 84 & 73 \\
LINE6 & 68 & 52 & 133 & 110 & 328 & 234 & 396 & 220 & 70 \\
CIRC6 & 65 & 120 & 329 & 362 & 211 & 227 & 353 & 175 & 252 \\
CON6 & 133 & 133 & 545 & 176 & 96 & 172 & 307 & 168 & 117 \\
\hline
NL8 \mirrored{} & 27 & 14 & 24 & 24 & 34 & 22 & 79 & 20 & 13 \\
SUP8 \mirrored{} & 24 & 19 & 30 & 25 & 37 & 32 & 51 & 16 & 19 \\
GAL8 \mirrored{} & 27 & 17 & 30 & 23 & 31 & 17 & 27 & 25 & 17 \\
INCR8 \mirrored{} & 23 & 17 & 26 & 33 & 29 & 19 & 75 & 21 & 16 \\
LINE8 \mirrored{} & 22 & 16 & 49 & 26 & 89 & 29 & 83 & 35 & 15 \\
CIRC8 \mirrored{} & 25 & 16 & 117 & 84 & 128 & 49 & 120 & 47 & 14 \\
CON8 \mirrored{} & 37 & 9 & 92 & 49 & 14 & 12 & 45 & 41 & 43 \\
NL8 & 12 & 8 & 11 & 10 & 15 & 16 & 27 & 12 & 6 \\
SUP8 & 10 & 6 & 12 & 10 & 35 & 16 & 31 & 23 & 8 \\
GAL8 & 11 & 10 & 11 & 9 & 43 & 16 & 26 & 23 & 7 \\
INCR8 & 11 & 9 & 10 & 9 & 33 & 14 & 29 & 20 & 4 \\
LINE8 & 10 & 10 & 11 & 12 & 18 & 15 & 43 & 20 & 6 \\
CIRC8 & 12 & 7 & 22 & 17 & 27 & 18 & 47 & 12 & 10 \\
CON8 & 22 & 10 & 40 & 22 & 51 & 34 & 65 & 22 & 9 \\
    \end{tabular}
    }}%
  \end{center}
\end{table}

\paragraph{IP bounds.}
We now investigate the bounds that the IP solver can obtain, again compared to the best known solution values.
It turned out that within the time limit of 1\, hour all configurations could solve the instances with $n=4$.
Hence, in \cref{tab_ip_bounds} we only report about the instances with $n \in \{6,8\}$.

We first observe that almost all configurations could solve the mirrored instances with $n=6$.
One exception stands out, namely the mirrored version of GAL6 which could be solved without any strengthening, but could not be solved when adding~\eqref{eq_homestand_roadtrip_flow}.
We inspected the computations more closesly, but could not find a proper reason: in both runs, Gurobi generated between 3000 and 3500 cutting planes of roughly the same types.
The bound at the end of the root node was also much worse for the successful run, which was expected because inequalities~\eqref{eq_homestand_roadtrip_flow} significantly improve the bound.
Moreover, in the unsuccessful run, only slightly fewer branch-and-bound nodes were processed.
Hence, we can only guess that the reason might be a badly chosen branching strategy which was triggered once these cutting planes were added.

Let us turn to the other instances.
Here, one observes that the addition or removal of~\eqref{eq_travel_away_away_lifted}--\eqref{eq_travel_away_home_lifted} does not have much impact on the bound after executing 1\,hour of branch-and-cut.
The same holds for the addition of the flow inequalities~\eqref{eq_flow_distinct} or the flow equations~\eqref{eq_flow_equations} for $i \neq t$ or the addition of the flow inequalities~\eqref{eq_flow_distinct} for $i=t$ or their strengthened version, the home-flow inequalities~\eqref{eq_flow_home}.
Moreover, the positive impact of \eqref{eq_flow_distinct} or of~\eqref{eq_homestand_roadtrip_flow} manifests itself also in the quality of the bounds obtained during branch-and-bound.
By comparing, for $n=8$, the bounds from \cref{tab_ip_bounds} with those in \cref{tab_lp_bounds} one observes that the numbers do not differ much.
This means that one hour of computation time does not close much integrality gap beyond that of the initial LP relaxation.
This gives rise to the question whether this may be due to large processing times per branch-and-bound node.

\paragraph{Number of branch-and-bound nodes.}
Consider \cref{tab_ip_nodes} where the number of processed nodes is depicted.
We observe that in some cases it takes almost a minute to process a branch-and-bound node, which however still yields several thousand processed nodes in total.
However, the bound improvement due to branching is apparently rather low.
In addition, we can easily see that constraints~\eqref{eq_travel_away_away_lifted}--\eqref{eq_travel_away_home_lifted} reduce the node throughput significantly, which yields a worse performance in terms of closed gap.
The most likely reason for this slowdown is the sheer amount of such cutting planes (see \cref{tab_instance_characteristics}).
This clearly explains why the bounds obtained after one hour by the full model without~\eqref{eq_travel_away_away_lifted}--\eqref{eq_travel_away_home_lifted} are better than if one includes these constraints.

\paragraph{Conclusions.}
We conclude our work by observing that, from a practical point of view, adding strong inequalities to the base model pays off only in some cases.
In particular, the addition of inequalities that are lifted versions of already existing inequalities (whose faces only have a slightly higher dimension) is not helpful.
In contrast to this, new inequalities, i.e., those for which no similar inequality already exists, contribute a lot to the closed gap.
Moreover, since the processing of branch-and-bound nodes takes a nontrivial amount of time for our model, branch-and-bound itself does not suffice to close a lot of remaining gap.
This means that one must first improve the LP relaxation in a different fashion, e.g., via reformulation approaches (see, e.g.,~\cite{EastonNT03}) or by finding more classes of cutting planes.
Since our lifted inequalities~\eqref{eq_travel_away_away_lifted} slowed down the node processing times, we also propose to investigate whether the lazy generation of constraints~\eqref{model_basic_travel_away_away} could improve the node throughput.

\paragraph{Acknowledgements.}
The authors are grateful to the anonymous reviewer whose comments led to improvements of this manuscript, in particular for suggesting more extensive computational investigations.

\bibliographystyle{plain}
\bibliography{traveling-tournament}

\appendix
\section{Tournaments for facet proofs}
\label{sec_tournaments}
 
\subsection{Tournaments for Theorem~\ref*{thm_dim}}

\thmDimNoTravel*

\begin{proof}
  If $t = i$, let $i' \in V \setminus \{i,j,t\}$ and $j' \coloneqq j$.
  Otherwise, let $i' \coloneqq i$ and $j' \in V \setminus \{i,j,t\}$.
  Note that in either case $i'$, $j'$ and $t$ are distinct.
  We construct tournament $T$ from a canonical factorization by permuting slots and teams such that $(1,i',t), (2,j',t) \in T$.
  Hence, team $t$ travels from venue $i'$ to venue $j'$, which implies that team $t$ never travels from venue $i$ to venue $j$ since exactly one of the teams $i'$, $j'$ is equal to its counterpart $i$, $j$.
\end{proof}

\thmDimHomeAway*

\begin{proof}
  We construct tournament $T$ from a canonical factorization by permuting slots and teams such that $(1,i,j), (k,j,i) \in T$.
  Tournament $T'$ is obtained from $T$ by~\eqrefHomeAwaySwapDefault{}.
\end{proof}

\thmDimPartialSlot*

\begin{proof}
  Let $M_\ell$ for all $\ell \in S$ be the perfect matchings of the canonical factorization.
  By permuting teams we can assume $\{i,j\}, \{i',j'\} \in M_1 = M_n$.
  We now exchange the roles of teams $j$ and $j'$ only in perfect matchings $M_n$, $M_{n+1}$, \dots, $M_{2n-2}$, which maintains the property that each edge appears in exactly two perfect matchings.
  Tournament $T$ is obtained by orienting the edges in a complementary fashion and permuting slots such that
  $(1,i,j), (1,i',j'), (k,i,j'), (k,i',j) \in T$.

  Finally, tournament $T'$ is constructed from $T$ by~\eqrefPartialSlotSwapDefault{}.
\end{proof}

\subsection{Tournaments for Theorem~\ref*{thm_nonneg}}

For the claims in the proof of \cref{thm_nonneg}, we are given a particular match $m^\star = (k^\star, i^\star, j^\star) = (2,4,3) \in \allMatches$.

\thmNonnegNoTravel*

\begin{proof}
  Let tournament $T$ be as constructed in the proof of \cref{thm_dim_no_travel}.
  If $m^\star \in T$, we apply a cyclic permutation of the slots, mapping slot $k$ to $k+1$ for $k \in S \setminus \{2n-2\}$ and slot $2n-2$ to $1$.
  This preserves the second requirement and establishes $m^\star \notin T$.
\end{proof}

\thmNonnegHomeAway*

\begin{proof}
  We construct tournament $T''$ from a canonical factorization by permuting slots and teams such that $(1,i,j), (k,j,i) \in T''$.
  If $m^\star \notin T''$, let $T \coloneqq T''$.

  Otherwise, if $k^\star \neq k$, then let $k' \in S$ be such that $(k',j^\star,i^\star) \in T''$.
  Tournament $T$ is obtained from $T''$ by~\eqrefHomeAwaySwap{$k^\star$}{$k'$}{$i^\star$}{$j^\star$}.
  Due to $k^\star \neq k$ and $k^\star \neq 1$, we have $(1,i,j), (k,j,i) \in T''$, but $m^\star \notin T$, and hence $T$ satisfies all requirements.

  Otherwise, $k^\star = k$ and $\{i,j\} \neq \{i^\star,j^\star\}$ hold.
  Together with $(k,j,i),(k^\star,i^\star,j^\star) \in T''$ this implies that $i,j,i^\star,j^\star$ must be distinct.
  Again, let $k' \in S$ be such that $(k',j^\star,i^\star) \in T''$ and construct $T$ from $T''$ by~\eqrefHomeAwaySwap{$k^\star$}{$k'$}{$i^\star$}{$j^\star$}.
  Since $i,j,i^\star,j^\star$ are distinct, we also have $(1,i,j), (k,j,i) \in T$, but $m^\star \notin T$ in this case, and hence $T$ satisfies all requirements.

  Finally, tournament $T'$ is obtained from $T$ by~\eqrefHomeAwaySwapDefault{}.
\end{proof}

\thmNonnegPartialSlot*

\begin{proof}
  Let $T$ be as constructed in the proof of \cref{thm_dim_partial_slot}.
%
  If $(k^\star,i^\star,j^\star) \in T$, let $k' \in V$ be such that $(k',j^\star,i^\star) \in T_1$ and modify $T$ via a home-away swap~\eqrefHomeAwaySwap{$k^\star$}{$k'$}{$i^\star$}{$j^\star$}.
  By assumptions on $k^\star$, $i^\star$ and $j^\star$, this operation does not affect the matches $(1,i,j), (1,i',j'), (k,i,j'), (k,i',j)$ above, i.e., these remain in $T$.
  However, after the modification, we have $(k^\star,i^\star,j^\star) \notin T$.

  Finally, tournament $T'$ is constructed from $T$ by~\eqrefPartialSlotSwapDefault{}.
\end{proof}

\subsection{Tournaments for Theorem~\ref*{thm_travel_away_away_lifted}}

For the claims in the proof of \cref{thm_travel_away_away_lifted}, we are given a slot $k^\star \in 
\{n,n+1,\dotsc,2n-3\}$ and three distinct teams $t^\star$, $i^\star$ and $j^\star$.
Note that we also assume $n \geq 6$.
To enhance readability of the proofs we restate the claim that contains sufficient conditions for satisfying~\eqref{model_basic_travel_away_away} with equality.

\thmTravelAwayAwayLiftedFace*

\thmTravelAwayAwayLiftedNoTravel*

\begin{proof}
  If $t = i$, let $i' \in V \setminus \{i,j,t\}$ and $j' \coloneqq j$.
  Otherwise, let $i' \coloneqq i$ and $j' \in V \setminus \{i,j,t\}$.
  Note that in either case $i'$, $j'$ and $t$ are distinct.
  We distinguish three cases.

  \textbf{Case 1: $t^\star \neq t$ or $i^\star \notin \{i',j'\}$.}
  We construct tournament $T$ from a canonical factorization by permuting slots and teams such that
  $(1,i',t), (2,j',t), (k^\star,i^\star,t^\star) \in T$ holds and such that $t^\star$ plays away in slot $k^\star+1$.
  Hence, team $t$ travels from venue $i'$ to venue $j'$, which implies that team $t$ never travels from venue $i$ to venue $j$ since exactly one of the teams $i'$, $j'$ is equal to its counterpart $i$, $j$.

  \textbf{Case 2: $t^\star = t$ and $i^\star = i'$.}
  We construct tournament $T$ from a canonical factorization by permuting slots and teams such that
  $(k^\star,i^\star,t^\star), (k^\star+1,j',t^\star) \in T$ holds.

  \textbf{Case 3: $t^\star = t$ and $i^\star = j'$.}
  We construct tournament $T$ from a canonical factorization by permuting slots and teams such that
  $(k^\star,i',t^\star), (k^\star,i^\star,t^\star) \in T$ holds and such that $t^\star$ plays away in slot $k^\star+1$.

  In all cases team $t$ travels from venue $i'$ to venue $j'$, which implies that team $t$ never travels from venue $i$ to venue $j$ since exactly one of the teams $i'$, $j'$ is equal to its counterpart $i$, $j$.
  Moreover, $(k^\star,i^\star,t^\star) \in T$ holds and team $t^\star$ plays away in slot $k^\star+1$, which concludes the proof.
\end{proof}

\thmTravelAwayAwayLiftedHomeAway*

\begin{proof}
  Let $M_\ell$ for all $\ell \in S$ be the perfect matchings of the canonical factorization.
  We distinguish ten cases:

\medskip
  \textbf{Case 1: $\{i,j\} \cap \{j^\star, t^\star\} = \varnothing$.}
  We permute slots such that $\{i,j\} \in M_1, M_k$ and such that there are edges $\{t',i'\} \in M_{k^\star}$ and $\{t',j'\} \in M_{k^\star+1}$ with distinct $t',i',j' \in V \setminus \{i,j\}$.
  Then we permute the teams $V \setminus \{i,j\}$ such that $i'$ is mapped to some team $i^{\#} \neq i^\star$, $j'$ is mapped to $j^\star$ and $t'$ is mapped to $t^\star$.
  Tournament $T$ is obtained by orienting the matching edges in a complementary fashion such that $(1,i,j), (k,j,i), (k^\star,i^{\#},t^\star), (k^\star+1,j^\star,t^\star) \in T$ hold.
  Note that team $t^\star$ travels from venue $i^{\#}$ to venue $j^\star$ and thus never from venue $i^\star$ to venue $j^\star$.

\medskip
  \textbf{Case 2: $\{i,j\} \cap \{i^\star, t^\star\} = \varnothing$.}
  We permute slots such that $\{i,j\} \in M_1, M_k$ and such that there are edges $\{t',i'\} \in M_{k^\star}$ and $\{t',j'\} \in M_{k^\star+1}$ with distinct $t',i',j' \in V \setminus \{i,j\}$.
  Then we permute the teams $V \setminus \{i,j\}$ such that $i'$ is mapped to $i^\star$, $j'$ is mapped to some team $j^{\#} \neq j^\star$ and $t'$ is mapped to $t^\star$.
  Tournament $T$ is obtained by orienting the matching edges in a complementary fashion such that $(1,i,j), (k,j,i), (k^\star,i^\star,t^\star), (k^\star+1,j^{\#},t^\star) \in T$ hold.
  Note that team $t^\star$ travels from venue $i^\star$ to venue $j^{\#}$ and thus never from venue $i^\star$ to venue $j^\star$.

\medskip
  \textbf{Case 3: $\{i,j\} = \{i^\star,j^\star\}$ and $k = k^\star$.}
  We construct tournament $T$ by permuting slots such that $(1,i,j), (k,j,i), (k^\star+1,j^\star,t^\star) \in T$ holds.
  Note that team $t^\star$ plays against a team in $V \setminus \{i^\star,j^\star\}$ before traveling to venue $j^\star$ and thus never travels from venue $i^\star$ to venue $j^\star$.

\medskip
  \textbf{Case 4: $\{i,j\} = \{i^\star,j^\star\}$ and $k \neq k^\star$.}
  We permute slots such that $\{i,j\} \in M_1, M_k$, $\{i^\star,t^\star\} \in M_{k^\star}$ and $\{j',t^\star\} \in M_{k^\star+1}$ holds for some team $j' \neq j^\star$.
  Tournament $T$ is obtained by orienting the matching edges in a complementary fashion such that $(1,i,j), (k,j,i), (k^\star,i^\star,t^\star), (k^\star,j',t^\star) \in T$ holds.
  Note that team $t^\star$ travels from venue $i^\star$ to venue $j' \neq j^\star$ and thus never travels from venue $i^\star$ to venue $j^\star$.

\medskip
  \textbf{Case 5: $\{i,j\} \cap \{i^\star,j^\star,t^\star\} = \{t^\star\}$ and $k = k^\star$.}
  We construct tournament $T$ by permuting slots such that $(1,i,j), (k,j,i), (k^\star+1,j^\star,t^\star) \in T$ holds.
  Note that team $t^\star$ travels from its home venue or some venue different from $i^\star$ to venue $j^\star$ and thus never from venue $i^\star$ to venue $j^\star$.

\medskip
  \textbf{Case 6: $\{i,j\} \cap \{i^\star,j^\star,t^\star\} = \{t^\star\}$ and $k = k^\star+1$.}
  We construct tournament $T$ by permuting slots such that $(1,i,j), (k,j,i), (k^\star,i^\star,t^\star) \in T$ holds.
  Note that team $t^\star$ travels from $i^\star$ to its home venue or to some venue different from $j^\star$ and thus never from venue $i^\star$ to venue $j^\star$.

\medskip
  \textbf{Case 7: $\{i,j\} \cap \{i^\star,j^\star,t^\star\} = \{t^\star\}$ and $k \notin \{k^\star, k^\star+1\}$.}
  We construct tournament $T$ by permuting slots such that $(1,i,j), (k,j,i), (k^\star,j^\star,t^\star), (k^\star+1,i^\star,t^\star) \in T$ holds.
  Note that team $t^\star$ travels from venue $j^\star$ to venue $i^\star$ and thus never from venue $i^\star$ to venue $j^\star$.

\medskip
  \textbf{Case 8: $\{i,j\} = \{i^\star,t^\star\}$ and $k = k^\star+1$.}
  We construct tournament $T$ by permuting slots such that $(1,i,j), (k,j,i), (k^\star,j^\star,t^\star) \in T$ and such that $(k^\star-1,i^\star,t^\star) \notin T$ holds.
  Due to the last condition, team $t^\star$ never travels from venue $i^\star$ to venue $j^\star$.

\medskip
  \textbf{Case 9: $\{i,j\} = \{i^\star,t^\star\}$ and $k \notin \{k^\star,k^\star+1\}$.}
  We construct tournament $T$ by permuting slots such that $(1,i,j), (k,j,i), (k^\star+1,j^\star,t^\star) \in T$ holds.
  Since team $t^\star$ plays at venue $j^\star$ in slot $k^\star+1$ but does not play against $i^\star$ in slot $k^\star$ it never travels from venue $i^\star$ to venue $j^\star$.

\medskip
  \textbf{Case 10: $\{i,j\} = \{j^\star,t^\star\}$ and $k \notin \{k^\star,k^\star+1\}$.}
  We construct tournament $T$ by permuting slots such that $(1,i,j), (k,j,i), (k^\star,i^\star,t^\star) \in T$ holds.
  Since team $t^\star$ plays at venue $i^\star$ in slot $k^\star$ but does not play against $j^\star$ in slot $k^\star+1$ it never travels from venue $i^\star$ to venue $j^\star$.

\medskip
  It is easy to check that all allowed triples $(k,i,j)$ are covered by the cases.
  Moreover, in all cases tournament $T'$ is obtained from $T$ by~\eqrefHomeAwaySwapDefault{} and also satisfies the required properties.
  In particular, also in $T'$ team $t^\star$ does not travel from venue $i^\star$ to venue $j^\star$ since all previous arguments were symmetric in $i$ and $j$.
\end{proof}

\thmTravelAwayAwayLiftedPartialSlot*

\begin{proof}
  Denote by $I \coloneqq \{ (1,i,j), (1,i',j'), (1,i,j'), (1,i',j), (k,i,j), (k,i',j'), (k,i,j'), (k,i',j) \}$ the set of matches required for tournaments $T$ or $T'$.
  Let $M_\ell$ for all $\ell \in S$ be the perfect matchings of the canonical factorization.
  By permuting teams we can assume $\{i,j\}, \{i',j'\} \in M_1 = M_n$.
  We now exchange the roles of teams $j$ and $j'$ only in perfect matchings $M_n$, $M_{n+1}$, \dots, $M_{2n-2}$, which maintains the property that each edge appears in exactly two perfect matchings.
  Then we permute slots such that $\{i,j\}, \{i',j'\} \in M_1$ and $\{i,j'\}, \{i',j\} \in M_k$ hold.

  We first describe two constructions of tournament $T$ that are applicable in many cases.
  Note that via~\eqrefPartialSlotSwapDefault{}, also tournament $T'$ is determined.

\medskip
  \textbf{Case 1: $(i^\star,t^\star) \notin P$, $k \neq k^\star$ and $(k^\star+1,j^\star,t^\star) \notin I$.}
  We can assume that $\{i^\star,t^\star\} \notin M_1 \cup M_k$ holds since otherwise we have $\{i^\star,t^\star\} \cap \{i,j,i',j'\} = \varnothing$ due to $(i^\star,t^\star) \notin P$, which allows to permute teams in $V \setminus \{i,j,i',j'\}$ in order to avoid this situation.
  Hence, we can permute slots such that $\{i^\star,t^\star\} \in M_{k^\star}$ holds, which is possible due to $k^\star \neq k$ and $k^\star \geq n$.
  We can now orient the matching edges in a complementary fashion such that $(1,i,j), (1,i',j'), (k,i,j'), (k,i',j), (k^\star,i^\star,t^\star), (k^\star+1,j'',t^\star) \in T$ holds for some $j'' \neq j^\star$.
  The latter is possible due to $(k^\star+1,j^\star,t^\star) \notin I$.
  Since in both tournaments $T$ and $T'$, team $t^\star$ travels from venue $i^\star$ to venue $j'' \neq j^\star$, it never travels from venue $i^\star$ to venue $j^\star$.
  Hence, $T$ and $T'$ satisfy the requirements of the claim.

\medskip
  \textbf{Case 2: $(j^\star,t^\star) \notin P$, $k \neq k^\star+1$ and $(k^\star,i^\star,t^\star) \notin I$.}
  We can assume that $\{j^\star,t^\star\} \notin M_1 \cup M_k$ holds since otherwise we have $\{j^\star,t^\star\} \cap \{i,j,i',j'\} = \varnothing$ due to $(j^\star,t^\star) \notin P$, which allows to permute teams in $V \setminus \{i,j,i',j'\}$ in order to avoid this situation.
  Hence, we can permute slots such that $\{j^\star,t^\star\} \in M_{k^\star+1}$ holds, which is possible due to $k^\star+1 \neq k$ and $k^\star \geq n$.
  We can now orient the matching edges in a complementary fashion such that $(1,i,j), (1,i',j'), (k,i,j'), (k,i',j), (k^\star,i'',t^\star)$, $(k^\star+1,j^\star,t^\star) \in T$ holds for some $i'' \neq i^\star$.
  The latter is possible due to $(k^\star,i^\star,t^\star) \notin I$.
  Since in both tournaments $T$ and $T'$, team $t^\star$ travels from venue $i'' \neq i^\star$ to venue $j^\star$, it never travels from venue $i^\star$ to venue $j^\star$.
  Hence, $T$ and $T'$ satisfy the requirements of the claim.

\medskip
  If condition~\ref{thm_travel_away_away_lifted_partial_slot_out_out} of the claim is satisfied, then $(i^\star,t^\star) \notin P$ implies $(k^\star,i^\star,t^\star) \notin I$ and $(j^\star,t^\star) \notin P$ implies $(k^\star+1,j^\star,t^\star) \notin I$.
  Hence, depending on $k$, (at least) one of the two cases above is applicable and we are done.

\medskip
  If condition~\ref{thm_travel_away_away_lifted_partial_slot_out_in} of the claim is satisfied, then case~1 is applicable unless $(k^\star+1,j^\star,t^\star) \in I$ holds.
  However, this implies $k = k^\star+1$, which is excluded by condition~(ii).

\medskip
  If condition~\ref{thm_travel_away_away_lifted_partial_slot_in_out} of the claim is satisfied, then case~2 is applicable unless $k = k^\star+1$ or $(k^\star,i^\star,t^\star) \in I$ holds.
  However, the latter would imply $k = k^\star$, which is excluded by condition~(iii).
  Hence, $k = k^\star + 1$, $i^\star \in \{i,i'\}$, $t^\star \in \{j,j'\}$ and $j^\star \notin \{i,j,i',j'\}$ hold.
  We permute teams $V \setminus \{i,j,i',j'\}$ and slots $S \setminus \{1,k\}$ such that $\{j^\star,t^\star\} \in M_{k^\star}$ holds.
  We can now orient the matching edges in a complementary fashion such that $(1,i,j), (1,i',j'), (k,i,j'), (k,i',j), (k^\star,j^\star,t^\star) \in T$ holds.
  In tournaments $T$ and $T'$, team $t^\star$ plays away at venue $i^\star$ in slots $1$ or $k^\star+1$.
  Due to $k^\star \geq n$ we have $k^\star - 1 \neq 1$, and thus team $t^\star$ travels from a venue different from $i^\star$ to venue $j^\star$, and thus never travels from venue $i^\star$ to venue $j^\star$.

\medskip
  If condition~\ref{thm_travel_away_away_lifted_partial_slot_in_in} of the claim is satisfied, then $\{i,i'\} = \{i^\star,j^\star\}$ and $t^\star \in \{j,j'\}$ holds.
  We orient the matching edges in a complementary fashion such that $(1,i,j), (1,i',j'), (k,i,j'), (k,i',j) \in T$ holds.
  In tournaments $T$ and $T'$, team $t^\star$ plays away at venues $i^\star$ and $j^\star$ in slots $1$ and $k^\star \geq n$, and thus never travels from venue $i^\star$ to venue $j^\star$.
  Moreover, match $(k^\star,i^\star,t^\star)$ is contained in one tournament and match $(k^\star,j^\star,t^\star)$ in the other.
\end{proof}

\subsection{Tournaments for Theorem~\ref*{thm_travel_home_away}}

For the claims in the proof of \cref{thm_travel_home_away}, we are given a slot $k^\star \in S \setminus \{2n-2\}$ and two distinct teams $t^\star, j^\star$.
Note that we also assume $n \geq 8$.
To enhance readability of the proofs we restate the claim that contains sufficient conditions for satisfying~\eqref{eq_travel_home_away_lifted} with equality.

\thmTravelHomeAwayLiftedFace*

\thmTravelHomeAwayLiftedNoTravel*

\begin{proof}
  We distinguish two cases.

\medskip
  \textbf{Case 1: $(t,i) \neq (t^\star,j^\star)$ or $k^\star \leq 2n-4$.}
  If $t = i$, let $i' \in V \setminus \{i,j,t,t^\star,j^\star\}$ and $j' \coloneqq j$.
  If $t \neq i$, let $i' \coloneqq i$ and $j' \in V \setminus \{i,j,t,t^\star,j^\star\}$.
  Note that in either case $i'$, $j'$ and $t$ are distinct.
  Now observe that $(j^\star,t^\star) \neq (j',t)$ holds since otherwise $j'=j$, and thus $t = i$ would hold, contradicting $(t,i,j) \neq (t^\star,t^\star,j^\star)$.
  We construct tournament $T$ from a canonical factorization by permuting slots and teams such that $(k^\star+1,j^\star,t^\star)$, $(k,i',t)$, $(k+1,j',t) \in T$ holds for $k = k^\star+1$ if $(i',t) = (j^\star,t^\star)$ and for some $k \in S \setminus \{k^\star,k^\star+1,2n-2\}$ otherwise.
  In the former case, since the first two matches $(k^\star+1,j^\star,t^\star)$ and $(k,i',t)$ are equal and $k+1 = k^\star + 2 \leq 2n-2$ holds, the existence of $T$ is obvious.
  In the latter case, the three distinct matches $(j^\star,t^\star)$, $(i',t)$ and $(j',t)$ have two scheduled in different slots.
  Since $t$ appears in two of the matches, edges $\{i',t\}$ and $\{j',t\}$ already appear in different matchings of a canonical factorization, and thus only slots must be permuted to construct $T$.
  It is easy to see that $T$ satisfies either condition~\ref{thm_travel_home_away_lifted_face_home_second} or~\ref{thm_travel_home_away_lifted_face_away_second} of \cref{thm_travel_home_away_lifted_face}.

\medskip
  \textbf{Case 2: $(t,i) = (t^\star,j^\star)$ and $k^\star = 2n-3$.}
  We construct tournament $T$ from a canonical factorization by permuting slots such that $(k^\star,j^\star,t^\star), (k^\star+1,i',t^\star) \in T$ holds for some $i' \in V \setminus \{j,t^\star,j^\star\}$, and such that $t^\star$ plays away in slot $k^\star-1$.
  In this case $T$ satisfies condition~\ref{thm_travel_home_away_lifted_face_away_first} of \cref{thm_travel_home_away_lifted_face}.
\end{proof}

\thmTravelHomeAwayLiftedHomeAway*

\begin{proof}
  We distinguish two cases.
  Note that via~\eqrefHomeAwaySwap{$\bar{k}$}{$k$}{$i$}{$j$}, also tournament $T'$ is determined.

\medskip
  \textbf{Case 1: $k \neq k^\star+1$.}
  We construct tournament $T$ from a canonical factorization by permuting slots and teams such that $(\bar{k},i,j)$, $(k,j,i), (k^\star+1,j^\star,t^\star) \in T$ holds.
  Since $\bar{k} \neq k^\star$ holds and $k = k^\star$ implies $t^\star \notin \{i,j\}$, tournaments $T$ and $T'$ both satisfy condition~\ref{thm_travel_home_away_lifted_face_home_second} or both satisfy condition~\ref{thm_travel_home_away_lifted_face_away_second} of \cref{thm_travel_home_away_lifted_face}.

\medskip
  \textbf{Case 2: $k = k^\star+1$.}
  We construct tournament $T$ from a canonical factorization by permuting slots and teams such that $(\bar{k},i,j)$, $(k,j,i), (1,j^\star,t^\star) \in T$ holds and, if $k^\star \neq 1$, team $t^\star$ plays at home in slot $k^\star$.
  Tournaments $T$ and $T'$ satisfy condition~\ref{thm_travel_home_away_lifted_face_init_first} (resp.\ condition~\ref{thm_travel_home_away_lifted_face_init_home} if $k^\star \neq 1$) of \cref{thm_travel_home_away_lifted_face}.
\end{proof}

\thmTravelHomeAwayLiftedPartialSlot*

\begin{proof}
  Let $M_\ell$ for all $\ell \in S$ be the perfect matchings of the canonical factorization.
  By permuting teams we can assume $\{i,j\}, \{i',j'\} \in M_1 = M_n$.
  We now exchange the roles of teams $j$ and $j'$ only in perfect matchings $M_n$, $M_{n+1}$, \dots, $M_{2n-2}$, which maintains the property that each edge appears in exactly two perfect matchings.
  Then we permute slots such that $\{i,j\}, \{i',j'\} \in M_{\bar{k}}$ and $\{i,j'\}, \{i',j\} \in M_{k}$ hold.
  We describe how to construct tournament $T$.
  Note that via~\eqrefHomeAwaySwap{$\bar{k}$}{$k$}{$i$}{$j$}, also tournament $T'$ is determined.
  We distinguish three cases.

\medskip
  \textbf{Case 1: $k \neq k^\star+1$ and $(j^\star,t^\star) \notin \{ (i,j), (i',j'), (i,j'), (i',j) \}$.}
  We can permute teams $V \setminus \{i,j,i',j'\}$ and slots $S \setminus \{\bar{k}, k\}$ such that $\{j^\star,t^\star\} \in M_{k^\star+1}$ holds.
  Tournament $T$ is obtained by orienting the matching edges in a complementary fashion such that $(\bar{k},i,j)$, $(\bar{k},i',j')$, $(k,i,j')$, $(k,i',j)$, $(k^\star+1,j^\star,t^\star) \in T$ holds.
  Since a partial slot swap does not change the home-away pattern of $t^\star$, tournaments $T$ and $T'$ both satisfy condition~\ref{thm_travel_home_away_lifted_face_home_second} or both satisfy condition~\ref{thm_travel_home_away_lifted_face_away_second} of \cref{thm_travel_home_away_lifted_face}.

\medskip
  \textbf{Case 2: $k = k^\star+1$ and $(j^\star,t^\star) \notin \{ (i,j), (i',j'), (i,j'), (i',j) \}$.}
  We can permute teams $V \setminus \{i,j,i',j'\}$ and slots $S \setminus \{\bar{k}, k\}$ such that $\{j^\star,t^\star\} \in M_1$ holds.
  Tournament $T$ is obtained by orienting the matching edges in a complementary fashion such that $(\bar{k},i,j)$, $(\bar{k},i',j')$, $(k,i,j')$, $(k,i',j)$, $(1,j^\star,t^\star) \in T$ holds and, if $k^\star \neq 1$, such that team $t^\star$ plays at home in slot $k^\star$.
  Since a partial slot swap does not change the home-away pattern of $t^\star$, tournaments $T$ and $T'$ both satisfy condition~\ref{thm_travel_home_away_lifted_face_init_first} or both satisfy condition~\ref{thm_travel_home_away_lifted_face_init_home} of \cref{thm_travel_home_away_lifted_face}.

\medskip
  \textbf{Case 3: $k \notin \{1,k^\star,k^\star+1\}$ and $(j^\star,t^\star) \in \{ (i,j), (i',j'), (i,j'), (i',j) \}$.}
  Let $k' \in S \setminus \{1,k^\star-1,k^\star,k^\star+1,\bar{k}-1,\bar{k},k-1,k\}$ (note that $n \geq 8$ implies $|S| \geq 14$).
  We can permute teams $V \setminus \{i,j,i',j'\}$ and slots $S \setminus \{\bar{k}, k\}$ such that $\{j'',t^\star\} \in M_{k^\star}$, $\{i'',t^\star\} \in M_{k'-1}$ and $\{j^\star,t^\star\} \in M_{k'}$ hold for distinct teams $i'',j'' \in V \setminus \{i,j,i',j'\}$.
  Tournament $T$ is obtained by orienting the matching edges in a complementary fashion such that $(\bar{k},i,j)$, $(\bar{k},i',j')$, $(k,i,j')$, $(k,i',j)$, $(k^\star,t^\star,j'')$, $(k'-1,i'',t^\star)$, $(k',j^\star,t^\star) \in T$ holds.
  Both tournaments $T$ and $T'$ satisfy condition~\ref{thm_travel_home_away_lifted_face_home_only} of \cref{thm_travel_home_away_lifted_face} since team $t^\star$ travels from venue $i'' \neq t^\star$ to venue $j^\star$, and thus never travels from its home venue to venue $j^\star$.
\end{proof}

\subsection{Tournaments for Theorem~\ref*{thm_travel_first}}

For the claims in the proof of \cref{thm_travel_first} we are given a team $t^\star$ and a venue $j^\star \neq t^\star$.
Note that we also assume $n \geq 6$.
To enhance readability of the proofs we restate the claim that contains sufficient conditions for satisfying~\eqref{model_basic_travel_first} with equality.

\thmTravelFirstFace*

\thmTravelFirstNoTravel*

\begin{proof}
  If $t = i$, let $i' \in V \setminus \{i,j,t,t^\star,j^\star\}$ and $j' \coloneqq j$.
  If $t \neq i$, let $i' \coloneqq i$ and $j' \in V \setminus \{i,j,t,t^\star,j^\star\}$.
  Note that in either case $i'$, $j'$ and $t$ are distinct.
  Moreover, we have $(j',t) \neq (j^\star,t^\star)$ since otherwise $t = i$ and thus $(t,i,j) = (t^\star,t^\star,j^\star)$ would contradict the assumption of the claim.
  If $(i',t) = (j^\star,t^\star)$ holds, then we construct tournament $T$ from a canonical factorization by permuting slots and teams such that $(1,j^\star,t^\star), (2,j',t) \in T$ holds.
  Otherwise, we construct $T$ with $(1,j^\star,t^\star), (2,i',t), (3,j',t) \in T$.

  In both cases, tournament $T$ satisfies condition~\ref{thm_travel_first_face_yes} and team $t^\star$ travels from venue $i'$ to venue $j'$, which implies that team $t$ never travels from venue $i$ to venue $j$ since exactly one of the teams $i'$, $j'$ is equal to its counterpart $i$, $j$.
\end{proof}

\thmTravelFirstHomeAway*

\begin{proof}
  Note that via~\eqrefHomeAwaySwap{$n$}{$k$}{$i$}{$j$}, also tournament $T'$ is determined.
  We distinguish two cases.

\medskip
  \textbf{Case 1: $k = 1$.}
  Let $i \in V \setminus \{ j^\star, t^\star \}$.
  We construct tournament $T$ from a canonical factorization by permuting slots and teams such that $(n,i,j)$, $(1,j,i), (2,i,t^\star), (3,j^\star,t^\star) \in T$ holds.
  In both tournaments, team $t^\star$ travels from venue $i \neq t^\star$ to venue $j^\star$, which implies that team $t^\star$ never travels from its home venue to venue $j^\star$.
  Hence, $T$ and $T'$ both satisfy condition~\ref{thm_travel_first_face_no} of \cref{thm_travel_first_face}.

\medskip
  \textbf{Case 2: $k \geq 2$ or $\{j^\star,t^\star\} \cap \{i,j,i',j'\} = \varnothing$.}
  We construct tournament $T$ from a canonical factorization by permuting slots and teams such that $(n,i,j)$, $(k,j,i), (1,j^\star,t^\star) \in T$ holds.
  Tournaments $T$ and $T'$ both satisfy condition~\ref{thm_travel_first_face_yes} of \cref{thm_travel_first_face}.
\end{proof}

\thmTravelFirstPartialSlot*

\begin{proof}
  Let $M_\ell$ for all $\ell \in S$ be the perfect matchings of the canonical factorization.
  By permuting teams we can assume $\{i,j\}, \{i',j'\} \in M_1 = M_n$.
  We now exchange the roles of teams $j$ and $j'$ only in perfect matchings $M_n$, $M_{n+1}$, \dots, $M_{2n-2}$, which maintains the property that each edge appears in exactly two perfect matchings.
  Then we permute slots such that $\{i,j\}, \{i',j'\} \in M_n$ and $\{i,j'\}, \{i',j\} \in M_k$ hold.
  We describe how to construct tournament $T$.
  Note that via~\eqrefHomeAwaySwap{$n$}{$k$}{$i$}{$j$}, also tournament $T'$ is determined.
  We distinguish two cases.

\medskip
  \textbf{Case 1: $k = 1$.}
  We can permute teams $V \setminus \{i,j,i',j'\}$ and slots $S \setminus \{n, k\}$ such that $\{j^\star,t^\star\} \in M_3$ and $\{i'',t^\star\} \in M_2$ for some $i'' \in V \setminus \{ j^\star, t^\star \}$.
  Tournament $T$ is obtained by orienting the matching edges in a complementary fashion such that $(n,i,j)$, $(n,i',j')$, $(1,i,j')$, $(1,i',j)$, $(2,i'',t^\star), (3,j^\star,t^\star) \in T$ holds.
  Tournaments $T$ and $T'$ both satisfy condition~\ref{thm_travel_first_face_no} of \cref{thm_travel_first_face}.
  
\medskip
  \textbf{Case 2: $k \geq 2$ or $\{j^\star,t^\star\} \cap \{i,j,i',j'\} = \varnothing$.}
  We can permute teams $V \setminus \{i,j,i',j'\}$ and slots $S \setminus \{n, k\}$ such that $\{j^\star,t^\star\} \in M_1$.
  Tournament $T$ is obtained by orienting the matching edges in a complementary fashion such that $(n,i,j)$, $(n,i',j')$, $(k,i,j')$, $(k,i',j)$, $(1,j^\star,t^\star) \in T$ holds.
  Tournaments $T$ and $T'$ both satisfy condition~\ref{thm_travel_first_face_yes} of \cref{thm_travel_first_face}.

\medskip
  \textbf{Case 3: $k \geq 2$ and $(j^\star,t^\star) \in \{ (i,j), (i',j'), (i,j'), (i',j) \}$.}
  We can permute slots such that the edges that match $t^\star$ in $M_{k-1}$, $M_k$, $M_{n-1}$ and $M_n$ are different (unless $k = n \pm 1$ in which case $M_k = M_{n-1}$ or $M_{k-1} = M_n$ holds).
  Tournament $T$ is obtained by orienting the matching edges in a complementary fashion such that $(n,i,j)$, $(n,i',j')$, $(k,i,j')$, $(k,i',j) \in T$ holds and such that $t^\star$ plays away in slots $k-1$, $k$, $n-1$ and $n$.
  Hence, in none of the tournaments $T$ and $T'$, team $t^\star$ travels from its home venue to venue $j^\star$, which shows that $T$ and $T'$ both satisfy condition~\ref{thm_travel_first_face_no} of \cref{thm_travel_first_face}.
\end{proof}

\subsection{Tournaments for Theorem~\ref*{thm_flow_distinct}}

For the claims in the proof of \cref{thm_flow_distinct} we are given a team $t^\star$ and a venue $i^\star \neq t^\star$.
Note that we also assume $n \geq 8$.

\thmFlowOutDistinctNoTravel*

\begin{proof}
  Let $T$ be a tournament from~\cref{thm_dim_no_travel}.
  We do not need to restrict the schedule of team $t^\star$ since it leaves $i^\star \neq t^\star$ only once, namely after playing away against $i^\star$.
\end{proof}

\thmFlowOutDistinctHomeAway*

\begin{proof}
  We distinguish two cases:

  \textbf{Case 1: $\{i,j\} \neq \{i^\star,t^\star\}$.}
  Since $|S| = 2n - 2 \geq 6$ holds, there exists a slot $k^\star \in S \setminus \{1,k-1,k,2n-2\}$.
  We construct tournament $T$ from a canonical factorization by permuting slots and teams such that $(1,i,j), (k,j,i), (k^\star,i^\star,t^\star), (k^\star+1, i, t^\star) \in T$ for some $i \in V$.

  \textbf{Case 2: $\{i,j\} = \{i^\star,t^\star\}$.}
  We construct tournament $T$ from a canonical factorization by permuting slots and teams such that $(1,i,j), (k,j,i) \in T$.
  Moreover, $t^\star$ shall play at home in slots $2$ (unless $k=2$ and this conflicts with $(k,j,i) \in T$) and $k+1$ (unless $k+1 \notin S$).

  In both cases, tournament $T'$ is obtained from $T$ by~\eqrefHomeAwaySwapDefault{}.
  By construction, in both tournaments team $t^\star$ leaves $i^\star$ to its home venue $t^\star$, after slot $1$ in one tournament and after slot $k$ in the other tournament.
\end{proof}

\thmFlowOutDistinctPartialSlot*

\begin{proof}
  Let $M_\ell$ for all $\ell \in S$ be the perfect matchings of the canonical factorization.
  By permuting teams we can assume $\{i,j\}, \{i',j'\} \in M_1 = M_n$ and if both $i^\star$ and $t^\star$ are distinct from $i$, $i'$, $j$ and $j$', then $\{i^\star,t^\star\} \notin M_1 = M_n$ (this is possible because $n \geq 8$).
  We now exchange the roles of teams $j$ and $j'$ only in perfect matchings $M_n$, $M_{n+1}$, \dots, $M_{2n-2}$, which maintains the property that each edge appears in exactly two perfect matchings.
  Tournament $T$ is obtained by orienting the edges in a complementary fashion and permuting slots such that $(1,i,j), (1,i',j'), (k,i,j'), (k,i',j) \in T$.

  By construction and by the assumptions of the claim, $t^\star$ does not play away against $i^\star$ in slots $1$ or $k$.
  Hence, we can permute the slots in $S \setminus \{1,k\}$ such that for some slot $k' \in S \setminus \{1,k-1,k,2n-2\}$, we have $(k',i^\star,t^\star) \in T$ and such that $t^\star$ plays at home in slot $k' + 1$.

  Finally, tournament $T'$ is constructed from $T$ by~\eqrefPartialSlotSwapDefault{}.
  By construction, $T$ and $T'$ satisfy all requirements from the claim.
\end{proof}

\subsection{Tournaments for Theorem~\ref*{thm_flow_home_away}}

For the claims in the proof of \cref{thm_flow_home_away}, we are given a team $t^\star$ and two slots $k^\star,\bar{k} \in S$ with $k^\star \leq n-1$ and $k^\star < \bar{k} < k^\star+n-1$.
Note that we also assume $n \geq 6$.

\thmFlowOutHomeNoTravel*

\begin{proof}
  If $t = i$, let $i' \in V \setminus \{i,j,t,t^\star\}$ and $j' \coloneqq j$.
  Otherwise, let $i' \coloneqq i$ and $j' \in V \setminus \{i,j,t,t^\star\}$.
  Note that in either case $i'$, $j'$ and $t$ are distinct.
  We construct tournament $T$ from a canonical factorization by permuting slots and teams such that $(k,i',t), (k+1,j',t) \in T$ for some $k \in S \setminus \{2n-2\}$ and such that all away matches of $t^\star$ are in slots $2$, $3$, \dots, $n$.

  To see that this is possible, we discuss the cases in which $t^\star \in \{i',j',t\}$ holds.
  If $t^\star = i'$, then $t^\star = i \neq t,j'$ holds and we can choose $k \coloneqq 1$ such that in this slot team $t^\star$ plays at home against $t$.
  If $t^\star = j'$, then $t^\star = j$ and $t = i$ hold and we can choose $k \coloneqq n+1$ such that in this slot team $t^\star$ plays at home against $t$.
  Finally, if $t^\star = t$, then we can choose $k \coloneqq 2$ such that in slots $2$ and $3$ team $t^\star$ plays away against $i'$ and $j$'.
\end{proof}

\thmFlowOutHomeCofficients*

\begin{proof}
  We construct tournament $T$ from a canonical factorization by permuting slots such that all away matches of $t^\star$ are in slots $k$, $k+1$, \dots, $k+n-1$ in particular such that $(k,j,t^\star), (k+1,j',t^\star) \in T$ holds.
  Hence, in $T$, team $t^\star$ leaves its home venue exactly once to venue $j$.

  Finally, tournament $T'$ is constructed from $T$ via a home-away swap~\eqrefHomeAwaySwap{$k$}{$k+n-1$}{$t^\star$}{$j$}, which means that in $T'$ team $t^\star$ plays away in slots $k+1$, $k+2$, \dots, $k+n-1$, starting at venue $j'$ after playing at home in slot $k$.
\end{proof}

\thmFlowOutHomeHomeAway*

\begin{proof}
  We construct tournament $T$ from a canonical factorization by permuting slots and teams such that $(\bar{k},i,j), (k,j,i) \in T$ and such that team $t^\star$ plays away in consecutive matches $k'$, $k'+1$, \dots, $k'+n-1$ for some $k' \in \{1,2,\dotsc,n\}$.
  It is easy to see that such a slot $k'$ exists since we only have to make sure that $t^\star$ does not play against $i$ or $j$ in slots $\bar{k}$ and $k$.

  Finally, tournament $T'$ is constructed from $T$ via a home-away swap~\eqrefHomeAwaySwap{$\bar{k}$}{$k$}{$i$}{$j$}, which does not affect the home-away pattern of team $t^\star \notin \{i,j\}$.
\end{proof}

\thmFlowOutHomePartialSlot*

\begin{proof}
  Let $M_\ell$ for all $\ell \in S$ be the perfect matchings of the canonical factorization.
  By permuting teams we can assume $\{i,j\}, \{i',j'\} \in M_1 = M_n$.
  We now exchange the roles of teams $j$ and $j'$ only in perfect matchings $M_n$, $M_{n+1}$, \dots, $M_{2n-2}$, which maintains the property that each edge appears in exactly two perfect matchings.
  We distinguish two cases:
  
  \textbf{Case 1: $t^\star \in \{j,j'\}$.}
  By symmetry we can assume $t^\star = j$ without loss of generality.
  We construct tournament $T_1 $ by orienting the edges in a complementary fashion such that $T_1$ contains the matches $(1,i,j)$, $(1,i',j')$, $(n,i,j')$ and $(n,i',j)$ and such that $t^\star$ plays away in matches $1$, $2$, \dots, $n-1$, except for match $k' \in \{2,3,\dotsc,n-1\}$ in which $t^\star$ plays at home against $i'$ since the corresponding return match is scheduled in slot $n$.
  Hence, we have $(k',t^\star,i') \in T_1$.

  Tournament $T_2$ is now obtained from $T_1$ by exchanging slots $k'$ and $n$, i.e., $T_2$ contains matches $(1,i,j)$, $(1,i',j')$, $(k',i,j')$, $(k',i',j)$ and $(n,t^\star,i')$ and team $t^\star$ plays away in slots $1$, $2$, \dots, $n-1$.

  We construct tournament $T_3$ from $T_2$ via a cyclic shift by $s \coloneqq \min\{k_1-1,n-1\}$.
  We obtain $(1+s,i,j), (1+s,i',j'), (k'+s,i,j'), (k'+s,i',j) \in T_3$ and team $t^\star$ plays away in slots $1+s$, $2+s$, \dots, $n-1+s$.
  Note that $k_1,k_2 \in S' \coloneqq \{1+s,2+s, \dotsc, n-1+s\}$.

  Finally, we construct tournament $T$ from $T_3$ by exchanging slots within $S'$ such that $T$ contains the matches $(k_1,i,j)$, $(k_1,i',j')$, $(k_2,i,j')$ and $(k_2,i',j)$ while maintaining the property that team $t^\star$ plays away in slots $S'$.

  \textbf{Case 2: $t^\star \notin \{j,j'\}$.}
  We construct tournament $T_1 $ by orienting the edges in a complementary fashion such that $T_1$ contains the matches $(1,i,j)$, $(1,i',j')$, $(n,i,j')$ and $(n,i',j)$ and such that $t^\star$ plays away in matches $1$, $2$, \dots, $n-1$.

  Let $s \in \{0,1,2,\dotsc,n-1\}$ be such that $S' \coloneqq \{1+s, 2+s, \dotsc, n-1+s\}$ contains exactly one of the two slots $k_1, k_2$.
  By symmetry we can assume $k_1 \in S'$ and $k_2 \in S \setminus S'$ (otherwise exchange $k_1$ with $k_2$, $i$ with $i'$ and $j$ with $j'$).
  We construct tournament $T_2$ from $T_1$ via a cyclic shift by $s$.
  We obtain $(1+s,i,j), (1+s,i',j'), (n+s,i,j'), (n+s,i',j) \in T_2$ and team $t^\star$ plays away in slots $S'$.

  Finally, we construct tournament $T$ from $T_2$ by exchanging slot $1+s \in S'$ with $k_1 \in S'$ and slot $n+s \in S \setminus S'$ with $k_2 \in S \setminus S'$, which maintains the property that team $t^\star$ plays away in slots $S'$.
  Moreover, $(k_1,i,j), (k_1,i',j'), (k_2,i,j'), (k_2,i',j) \in T$ holds.

  In both cases we construct $T'$ from $T$ by a partial slot swap~\eqrefPartialSlotSwap{$k_1$}{$k_2$}{$i$}{$i'$}{$j$}{$j'$}.
\end{proof}

\thmFlowOutHomeHomeAwayInterval*

\begin{proof}
  Let $M_\ell$ for all $\ell \in S$ be the perfect matchings of the canonical factorization.
  We permute matchings such that $\{t^\star,j\} \in M_k,M_{k+1}$ and such that in matchings $M_{k^\star}$, $M_{k^\star+1}$, \dots, $M_{k-1}$, $M_k$, $M_{k+2}$, $M_{k+3}$, \dots, $M_{k^\star+n-2}$, $M_{k^\star+n-1}$ node $t^\star$ is matched to every other node exactly once.
  Tournament $T$ is constructed by orienting edge $\{t^\star,j\}$ as $(j,t^\star)$ in $M_k$ and as $(t^\star,j)$ in $M_{k+1}$, and such that team $t^\star$ plays away in slots $k^\star$, $k^\star+1$, \dots, $k^\star+n-1$ except for slot $k+1$.
  Team $t^\star$ leaves its home venue only before slot $k^\star$ and after slot $k+1$.
  We construct tournament $T'$ by~\eqrefHomeAwaySwap{$k$}{$k+1$}{$j$}{$t^\star$}, in which $t^\star$ plays at home in slot $k$ and away in slot $k+1$.
  Moreover, in $T'$ team $t^\star$ leaves its home venue only before slot $k^\star$ and after slot $k$.
\end{proof}

\thmFlowOutHomeHomeAwayMixed*

\begin{proof}
  Let $M_\ell$ for all $\ell \in S$ be the perfect matchings of the canonical factorization.
  We permute matchings such that $\{t^\star,j\} \in M_{k^\star},M_{k^\star+1}$, $\{t^\star,j'\} \in M_{k^\star+2}$ and such that in matchings $M_{k^\star+1}$, $M_{k^\star+2}$, \dots, $M_{k^\star+n-1}$ node $t^\star$ is matched to every other node exactly once.
  Tournament $T$ is constructed by orienting edge $\{t^\star,j\}$ as $(t^\star,j)$ in $M_{k^\star}$ and such that team $t^\star$ plays away in slots $k^\star+1$, $k^\star+2$, \dots, $k^\star+n-1$.
  Team $t^\star$ leaves its home venue exactly once, namely after slot $k^\star$ to venue $j$.
  We construct tournament $T'$ by~\eqrefHomeAwaySwap{$k^\star$}{$k^\star+1$}{$t^\star$}{$j$}, in which $t^\star$ plays away in slot $k^\star$ and home in slot $k^\star+1$.
  Moreover, in $T'$ team $t^\star$ leaves its home venue exactly twice, namely before slot $k^\star$ to venue $j$ and after slot $k^\star+1$ to venue $j'$.
\end{proof}

\end{document}